\documentclass[a4paper]{article} %% [11pt]{article}
\usepackage{amsmath,amsthm,amssymb}

\title{Quantization of classical integrable systems \\
Part III: \\
systems in $n$-dimensional Euclidean space}
\author{M. Marino and N. N. Nekhoroshev \\
{\small Dipartimento di Matematica, Universit\`{a} degli Studi di Milano,} \\
{\small via Saldini 50, I-20133 Milano (Italy)}}

\newtheorem{thm}{Theorem}%[section]
\newtheorem{cor}[thm]{Corollary}
\newtheorem{lem}[thm]{Lemma}
\newtheorem{prop}[thm]{Proposition}
\theoremstyle{definition}
\newtheorem{defn}{Definition}[section]

\theoremstyle{remark}
\newtheorem{rem}{Remark}[section]

\def\beq{\begin{equation}}
\def\eeq{\end{equation}}

\def\rank{{\rm rank}\,}
\def\ker{{\rm Ker}\,}
\def\Tr{{\rm Tr}\,}

\def\bop{{\cal B}}

\def\fop{{\cal F}}
\def\rn{\mathbb{R}}
\def\nn{\mathbb{N}}

\def\galg{\mathfrak{g}}
\def\gtyp{\mathfrak{g}_{\rm typ}}
\def\mtyp{M_{\rm typ}}

\def\pop{{\cal P}}

\def\oop{{\cal O}}

\def\ipn2{\left[\frac n2\right]}
\def\sspan{{\rm Span}\,}

\numberwithin{equation}{section} \numberwithin{thm}{section}

\begin{document}

\maketitle

\begin{abstract}
In this paper we give examples of applications of general methods
of quantization by symmetrization of classical integrable systems,
which have been illustrated in two previous works by the same
authors. We consider two classes of systems in $n$ spatial
dimensions, which respectively describe a point particle in a
central force field and a freely rotating rigid body. In the
former case, the application of the general methods to an
integrable classical system leads in an almost straightforward way
to the quasi-integrability of the corresponding quantum system. In
the latter case instead, a modification of the symmetrization
procedure is necessary in order to achieve quantum integrability
for $n= 6$.
\end{abstract}

\section{Introduction}
In two previous papers of this series, we have introduced the
concept of quasi-integrable quantum system \cite{part1}, and we
have established general methods to obtain examples of such
systems starting from classical integrable systems \cite{part2}.
Integrability of an operator $\hat H$ is defined as the existence
of a sufficiently large set $\hat F$ of operators which commute
with $\hat H$, more exactly a quasi-integrable set $\hat F$ of
operators. The main source of integrable sets are Lie closed sets
of operators commuting with $\hat H$. Making the union of several
such sets one can obtain an integrable set $\hat F$. An important
particular case of Lie closed set is a Lie algebra of operators.
General methods for the constructions of integrable sets are based
on the symmetrization of the products of operators, which
correspond to the elements of an integrable set of functions for
the classical system. In the present paper these methods will be
applied to some important classes of integrable classical and
quantum systems.

In section \ref{meth} we consider systems in a Euclidean space of
arbitrary dimension $n$, which describe a point particle in a
central force field. The discussion is then extended to more
general one-particle systems which are symmetric with respect to
the group of rotations $SO(n)$. For all these cases we construct
various types of classical integrable sets of functions. These
sets contain in general $2n-k$ elements, where $k$ is equal to the
number of elements in the central subset \cite{TMMO}. This number,
for the various integral sets here considered, can take all
possible values from 2 to $n$. Each integrable set can be applied
to all systems whose hamiltonian is an arbitrary function of the
central elements. In all these cases we show that, by applying the
general results of \cite{part2}, one can obtain a corresponding
integrable quantum systems with an equal number $k$ of central
operators.

In section \ref{meth2} we then consider a freely rotating rigid
body in a Euclidean space of arbitrary dimension $n$. In the
classical case, it is known that this system is completely
integrable. By applying our scheme of noncommutative integrability
\cite{TMMO, fomenko, fasso}, we find how the number $k$ of central
integrals depends on the space dimension $n$ and on the properties
of the set of generalized moments of inertia of the body. We then
show that, for $n\leq 5$, the application of the general results
of \cite{part2} leads in an almost straightforward way to the
integrability of the corresponding quantum systems. However, for
$n=6$ we find that a modification of the symmetrization procedure
is necessary in order to obtain a quasi-integrable set of
operators. In fact, a Manakov polynomial of fourth order in the
left-invariant momenta, which belongs to the central subset of the
classical integrable set, does not commute with the hamiltonian
operator after symmetrization in the momenta. However,
commutativity can be restored by adding to it a suitable second
order polynomial. One can conjecture that analogous procedures can
be applied also for $n>6$.

\section{One-particle systems in a $n$-dimensional
central force field} \label{meth}

It is well-known that several mechanical systems have been proven
to be integrable both at the classical and at the quantum level.
The integrability of various classes of systems is discussed for
instance in \cite{olsha0, olsha1, evans, etingof, tempesta,
Gra_Wint, marquette}. These systems usually consist of point
particles moving in a space of one or more dimensions, subjected
to an external potential or mutually interacting via suitable
two-particle potentials. In particular, in \cite{rodriguez} a
class of maximally superintegrable systems is studied, which
includes as a particular case the hydrogen atom in $n$ dimensions.
In this section we too shall consider systems in $n$ dimensions
which generalize in some sense the hydrogen atom problem, although
our aim will be partly different with respect to most of the cited
investigations. We shall not in fact restrict ourselves to
considering hamiltonians which are the sum of the usual kinetic
term and of a potential term dependent only on the position. We
shall instead consider a generic invariant hamiltonian with
respect to the group of $n$-dimensional rotations, and we shall
look for all the possible integrable sets of functions and
operators which can be constructed for such an hamiltonian. It is
then obvious that each of these sets can also be associated with
the entire class of integrable systems, whose hamiltonian is
expressible as a functions of the central elements of the set. In
this way our approach leads to the systematic individuation of
families of integrable systems in $n$ spatial dimensions. However,
we shall not discuss the possible physical interpretation of the
systems obtained with this method.

\subsection{Classical particle in a central force field} \label{ccff}

The hamiltonian function of this system in an $n$-dimensional
euclidean space has the following form:
\begin{equation}\label{claplace}
H= \frac 12 p^2 +U(r)
\end{equation}
where $U\in C^\infty (0,+\infty)$. Here we use the notation
\begin{equation*}
r= \sqrt {x^2}\,, \qquad  x= (x_1, \dots, x_n)\,, \qquad p = (p_1,
\dots, p_n)\,.
\end{equation*}
The first term on the right-hand side of (\ref{claplace})
corresponds to the kinetic energy of the particle, and the second
one corresponds to its potential energy. The configuration space
$K$ of this classical system is the $n$-dimensional euclidean
linear space ${\mathbb R}^n_x$, more exactly, $K= {\mathbb R}^n_x
\setminus \{0\}$. The group of orthogonal transformations $G=
SO(n)$ acts on this space. This action in an orthonormal basis in
${\mathbb R}^n_x$ is defined by orthogonal matrices. The dimension
of this Lie group is $N=n(n-1)/2$. The action of this Lie group
$G$ transfers onto the cotangent bundle $T^*{\mathbb R}^n_x =
{\mathbb R}^{2n}_{xp}$ to ${\mathbb R}^n_x$. This bundle, without
the cotangent space $T^*_0 {\mathbb R}^n_x$ to ${\mathbb R}^n_x$
at the point $0$, represents the phase space $M = T^* K$ of the
classical system. The action of this group $G$ conserves the
hamiltonian function $H= H(x,p) = \frac 12 p^2 +U(r)$ of the
classical system.

Let us denote the Lie algebra of the group $G$ as $\mathfrak{g}=
so(n)$. Each element $a$ of $\mathfrak{g}$ is associated with a
vector field $v_a$ on ${\mathbb R}^n_x$. The corresponding vector
field on the symplectic manifold $M \subset T^*{\mathbb R}^n_x$ is
hamiltonian with hamiltonian function $P_a(m)=\langle p, v_a(x)
\rangle$, where $p=m \in T^*_x {\mathbb R}^n_x$ is a linear form
on $T_x {\mathbb R}^n_x$. Let $(v_{a1}(x), \dots, v_{an}(x))$ be
the components of the vector field $v_a$ in coordinates $x=(x_1,
\dots, x_n)$; then $P_a(p,x)= \sum_i p_i v_{ai}(x)$, where $(p,x)$
are canonical coordinates on $T^*\mathbb{R}^n_x$. In the
considered case, in which the group is $G=SO(n)$, an element $a$
of the Lie algebra $\mathfrak{g}$ is represented by a
skew-symmetric matrix $A=A_a$, and the vector $v_a(x)$ at the
point $x=(x_1, \dots, x_n)$ has the form $v_a(x) = -A_a x$. The
action of the group $G$ on $T^*\mathbb{R}^n_x$ is a Poisson
action, i.e., $\{P_a, P_b\}= P_{[a,b]}$, where $a,b$ are any two
elements of $\mathfrak{g}$, and $[a,b]$ is the commutator defined
in this algebra. Each system of cartesian coordinates in
$\mathbb{R}^n_x$ defines in the algebra $\mathfrak{g}$ a basis
whose elements are given in these coordinates by matrices
$D^{ij}$, $1\leq i<j\leq n$, having a particularly simple form.
These matrices have only two non-zero elements which are equal to
$\pm 1$. Namely, $D^{ij}_{ij}= -D^{ij}_{ji}=1$, where
$D^{ij}_{kl}$ denotes the element of the matrix $D^{ij}$ lying at
the intersection of row $k$ and column $l$. Making use of
Kr\"onecker symbol $\delta_{ij}$, we can write in general
\begin{equation} \label{aij}
D^{ij}_{kl}= \delta_{ki} \delta_{lj} - \delta_{kj} \delta_{li} \,.
\end{equation}
According to this formula, a matrix $D^{ij}$ can naturally be
defined also for $i\geq j$, and we have $D^{ji} =-D^{ij}\
\forall\, i,j=1, \dots, n$. The commutation relations between
these matrices are
\begin{equation} \label{acomm}
[D^{ij}, D^{hk}] = -\delta_{ih} D^{jk} - \delta_{jk} D^{ih} +
\delta_{ik} D^{jh} + \delta_{jh} D^{ik} \,.
\end{equation}
Note that, whenever the commutator is nonzero, only one of the
four terms on the right-hand side is different from zero.

Let $P_{ij}$ denote the function $P_a$, corresponding to the
matrix $A_a =D^{ij}$. Then
\begin{equation} \label{pij}
P_{ij} = x_i p_j -x_j p_i\,.
\end{equation}
From (\ref{acomm}) one immediately derives that the Poisson
brackets relations between these functions are
\begin{equation} \label{ppoiss}
\{P_{ij}, P_{hk}\} = -\delta_{ih} P_{jk} - \delta_{jk} P_{ih} +
\delta_{ik} P_{jh} + \delta_{jh} P_{ik} \,.
\end{equation}
Let $P(x,p)$ denote the $N$-dimensional vector $(P_{ij}(x,p),\
1\leq i<j\leq n)$. It is easy to verify that the rank of the map
$P: \mathbb{R}^{2n}_{xp} \to \mathbb{R}^N$ at a typical point
$(x,p)$ is equal to $2n-3$. Let $P^2$ denote the square of the
length of vector $P$, i.e., $P^2:= \sum_{1\leq i<j \leq n}
P_{ij}^2=r^2p^2- (x\cdot p)^2$. Here $x\cdot p$ denotes the scalar
product of vectors $x$ and $p$, i.e., $x \cdot p := \sum_{i=1}^n
x_i p_i$. Using (\ref{ppoiss}) it is easy to check that
\begin{equation} \label{rotp}
\{P^2, P_{ij}\}=0
\end{equation}
for any component $P_{ij}$ of vector $P$. It is clear that one can
select $2n-4$ components $L=(P_{i_1 j_1}, \dots, P_{i_{2n-4}
i_{2n-4}})$ of this vector, such that the set $\Pi := (P^2, L)$
defines a regular map $\Pi: \mathbb{R}^{2n}_{xp} \to
\mathbb{R}^{2n-3}$ almost everywhere in $\mathbb{R}^{2n}_{xp}$.
This means that the rank of the map $\Pi$ is equal to $2n-3$
almost everywhere. A possible choice is $L = (P_{13}, P_{14},
\dots, P_{1n}, P_{23}, P_{24}, \dots, P_{2n})$.

Let us add to this set the hamiltonian function $H=H(x,p)$, and
denote by $F=(H, P^2;L)$ the resulting set of $2n-2$ functions. It
is easy to see that this set $F$ is functionally independent
almost everywhere, that is the set of critical points of the map
defined by this set has zero measure, is a closed set and is
nowhere dense. Using the relations
\begin{align}
\{x_i, P_{jk}\} &= \delta_{ij}x_k - \delta_{ik}x_j \,,\label{p1}\\
\{p_i, P_{jk}\} &= \delta_{ij}p_k - \delta_{ik}p_j \label{p2}
\end{align}
for $i,j,k=1, \dots, n$, it is also easy to verify that
\begin{equation}\label{rot}
\{r^2,P\}=0\,,\qquad \{p^2,P\} =0\,,
\end{equation}
that is $\{r^2,P_{jk}\}=\{p^2,P_{jk}\}=0$. Hence $\{H,P\}=0$.
Since $L \subset P$, this implies that the set $F$ has 2 central
functions, $H$ and $P^2$. According to the definition %\ref{clint},
given in \cite{part1}, $F$ is thus an integrable set with two
central integrals, and the system with hamiltonian function $H$ is
globally integrable with set of invariants $F$. The conservation
of $P$ implies that the orbit of the particle lies in a
2-dimensional plane.

%Let us recall the following well-known result.
\begin{prop} \label{strind2}
Let $(V_1,\ldots,V_l)$ and $(W_1,\ldots,W_s)$ be two sets of
functionally independent functions on a $2n$-dimensional
symplectic manifold, such that $\{V_i,W_k\} =0$ for $i=1,\ldots,l$
and $k=1,\ldots,s$. Then $l+s\leq 2n$.
\end{prop}
Taking into account the above well-known result, we can describe
the set of all integrable classical systems which are invariant
with respect to the action of the group $G=SO(n)$ on
$\mathbb{R}^{2n}_{xp}$, i.e., the integrable systems whose
hamiltonian $H$ is in involution with vector $P$.
\begin{lem} \label{40}
We have $\{H,P\}=0$ if and only if locally $H= f(p^2, r, P^2)$.
\end{lem}

\begin{proof}
Let us suppose that $H= f(p^2, r, P^2)$. Since the functions
$r,p^2,P^2$ are in involution with $P$ (see above), we have $\{H,
P\}= 0$.

Viceversa, let us suppose that there exists a function $H(x,p)$
such that $\{H,P\}=0$ and that has not locally the form $H=f(p^2,
r, P^2)$. In this case we would have 4 functionally independent
functions $p^2$, $r$, $P^2$ and $H$, which are in involution with
the $2n-3$ functionally independent functions of the set $\Pi$.
Since the sum of the numbers of functions belonging to these two
sets equals $2n+1$, this would be in contradiction with
proposition \ref{strind2}.
\end{proof}

\begin{prop} \label{cf1}
In the real analytic case, a hamiltonian function $H$ is in
involution with $P$ and is integrable if and only if it has
locally the following form: $H=f(p^2, r, P^2)$.
%???

The situation considered in all the present article refers to the
more general case of infinitely differentiable functions, that is
of class $C^\infty$. In this case, if $\{H,P\}=0$ and $H$ is
integrable, then $H=f(p^2, r, P^2)$. Viceversa, if $H=f(p^2, r,
P^2)$, where the 2-covector $\partial f/\partial (p^2, r)=
(\partial f/ \partial (p^2), \partial f/ \partial r)$ is not zero,
more exactly $\partial f/\partial (p^2, r) \neq 0$ almost
everywhere, then $H$ is integrable and $\{H,P\}=0$.
%Here the function
%$H(p,x)$ is defined for all $(p,x)$ such that $x\neq 0$ and
%$f$ is an arbitrary function of three variables. The domain of $f$
%lies in $\mathbb{R}^{3}_+$, that is in the main (positive) octant
%of $\mathbb{R}^{3}$.
Such systems are always integrable with $k=2$, with central
integrals $H$ and $P^2$. In certain cases it is possible to find
an additional integral, and to have integrability with $k=1$. For
example in the case of the Newton potential, that is for $H= p^2/2
-\alpha/r$, and in the case of identical uncoupled oscillators,
that is for $H= (p^2 +r^2)/2$. The system with hamiltonian
function $H$ which is a function only of the square angular
momentum, more exactly $H = f(P^2)$, where $df\neq 0$ almost
everywhere, is integrable with $k=1$.

Let us present the integrable sets $F= F(H)$ which correspond to
these hamiltonian functions. For $H=f(p^2, r, P^2)$, where the
function $f$ is not locally functionally dependent only on $P^2$,
one can take $F= (H, P^2; L)$, and therefore $k=2$. For $H= p^2/2
-\alpha/r$ one can take $F= (H; P^2, L, A_1)$, where $A_1=
\sum_{j=2}^n P_{1j}p_j - \alpha x_1/r$. For $H= (p^2 +r^2)/2$ one
can take $F= (H; H_1, H_2, \dots, H_{n-1},$ $P_{12}, P_{13},
\dots, P_{1n})$, where $H_i = (p_i^2 +x_i^2)/2$. For $H = f(P^2)$
one can take $F= (P^2;p^2, r, L)$. In all cases except the first
one, we have $k=1$.
\end{prop}

\begin{proof} According to the
previous lemma, $\{H, P\}= 0$ if and only if $H= f(p^2, r, P^2)$.

If $\partial f/\partial (p^2, r) \neq 0$ almost everywhere, we can
repeat the proof given at the beginning of section \ref{ccff} for
the case $H= p^2/2 -U(r)$. By adding the function $H$ to the set
$\Pi = (P^2, L)$, we thus obtain a set of $2n-2$ functions which
defines almost everywhere a regular map and has two central
integrals, $H$ and $P^2$. Therefore the system is integrable with
$k=2$. If instead $f=f(P^2)$, then the set $F= (P^2;p^2, r, L)$ is
an integrable set with $k=1$.

For the Kepler system, $H= p^2/2 -\alpha/r$, it is straightforward
to verify that $\{H, A_i\}=0$, where
\[
A_i=\sum_{j=1}^n P_{ij}p_j- \alpha \frac{x_i}r= \left(
p^2-\frac{\alpha}r\right)x_i- (x\cdot p)p_i
\]
for $i=1, \dots, n$ \cite{Landau}. Since vector $A$ lies in the
plane of the orbit and satisfies $A^2= 2P^2 H+ \alpha^2$,
obviously only one component of $A$ is functionally independent of
the set $(H, P^2, L)$. Hence $F=(H; P^2, L, A_1)$ is an integrable
set with $k=1$. Of course, one has to keep in mind that the
potential $U(x)= -\alpha/r$ is a solution of the Laplace equation
$\sum_{i=1}^n \partial^2 U/ \partial x_i^2 =0$ only for $n=3$.

The hamiltonian $H= (p^2 +r^2)/2$ actually describes a set of
resonators with equal frequencies, and will be considered again in
a following paper.

In all the considered cases, the linear independence of the
differentials of the functions $F(H)$ can be checked directly by
considering the corresponding jacobian matrices. In the analytic
case, either $H=f(P^2)$, or $\partial f/\partial (p^2, r) \neq 0$
almost everywhere: therefore integrability is always guaranteed.
\end{proof}

\begin{rem}
In the previous proposition we have considered the potential
$\alpha/r$ for a particle in a space of arbitrary dimension $n$.
Of course one has to keep in mind that such a potential is a Green
function for the $n$-dimensional Laplace operator only for $n=3$.
\end{rem}

Let us consider the case $H= f(p^2, r, P^2)$, with $\partial
f/\partial (p^2, r) \neq 0$ almost everywhere. We have seen that
these systems are integrable with $k=2$. This means that the
typical invariant surface for the phase-flow of the system is a
two-dimensional torus. It is however possible to find for the same
systems also integrable sets with a larger number $k$ of central
integrals, up to the maximum possible number $k=n$ which
corresponds to standard Liouville integrability. In this way one
can construct a larger class of integrable systems, which includes
all systems whose hamiltonian is an arbitrary function of the
central elements of the set. In general, such systems will no
longer be invariant under the action of the whole group $SO(n)$,
but only of some subgroup of it.

For example, in the familiar case $n=3$, one can take $F=(H, P^2;
P_{12}, P_{13})$, with $k=2$, but also $F=(H, P^2, P_{12})$, with
$k=3$, or in general $F=(H, P^2, P_{a})$, where $P_a$ is the
momentum associated with any arbitrary element $a\in so(3)$. It
follows that any system with hamiltonian $K= g(p^2, r, P^2,
P_{a})$, where $g$ is a function such that $\partial g/\partial
(p^2, r) \neq 0$, is integrable with $k=3$. Of course any such
system is only invariant with respect to the one-parameter
subgroup of the rotations around the axis associated with $a$.
\begin{defn}\label{feq}
We say that two integrable sets are {\it functionally equivalent}
if the elements of one set are locally functions of the elements
of the other set.
\end{defn}
Owing to the arbitrariness in the choice of the element $a\in
so(3)$, we see that, for the system with hamiltonian $H= f(p^2, r,
P^2)$, there exist infinitely many functionally inequivalent
integrable sets with $k=3$.

Let us now consider the case of arbitrary $n$. Note first of all
that in the integrable set described in proposition \ref{cf1}, all
$n$ coordinates of configuration space are treated on the same
footing, since all possible choices of the $2n-4$ noncentral
elements of the set actually lead to functionally equivalent
integrable sets. More generally, under any transformation of the
group $SO(n)$, the set $F$ with $k=2$ is transformed into an
equivalent set. One can however construct other integrable sets in
the following way. One takes the function $P^2$ as central
element, and then splits the set of $n$ coordinates of
configuration space into two arbitrary disjoint subsets. One takes
as additional central elements the two functions, one for each of
these two subsets, which are obtained by summing the squares of
all the components of $P$ acting on the coordinates of the subset.
Then, for any of the two subsets of coordinates, one can proceed
in two alternative ways. Either one takes as integral functions a
suitable set $L'$ of $2n'-4$ momenta acting on the coordinates of
the subset, where $n'$ is the number of such coordinates, or one
splits again the subset into two arbitrary smaller disjoint
subsets, and repeats the procedure. If one wishes, one can
continue splitting the subsets into two parts, until one is left
with only subsets consisting of either one or two space
coordinates (the splitting of a set of two coordinates is
ineffective with respect to the resulting integrable set). Coming
back to the case $n=3$, we see that the integrable system $F=(H,
P^2, P_{12})$ considered above corresponds to the splitting of the
set of coordinates $(x_1, x_2, x_3)$ into the two subsets $(x_1,
x_2)$ and $(x_3)$.

The general procedure is described in a formal way by the
following proposition.
\begin{prop}\label{bamb}
For any $n\geq 2$ and for any $z=1, \dots, n-1$, it is possible to
construct in a recursive manner sets $Z_{n,z}$ and $L_{n,z}$ of
polynomial functions of degree $\leq 2$ in the variables
$P_n:=(P_{ij}, 1\leq i<j \leq n)$, with the following properties:
\begin{enumerate}
\item $Z_{n,z}$ contains $z$ elements,

\item $L_{n,z}$ contains $2(n-z-1)$ elements, all of degree 1 in
$P_n$,

\item the set $\Pi_{n,z}:= (Z_{n,z}, L_{n,z})$ is functionally
independent,

\item $\{Z_{n,z}, \Pi_{n,z}\}=0$.
\end{enumerate}
\end{prop}
Given any set $A$, it is useful to denote with $\sharp A$ the
number of its elements. Properties 1 and 2 can thus be written
$\sharp Z_{n,z}= z$ and $\sharp L_{n,z}= 2(n-z-1)$ respectively.
Property 4 means that any element of $Z_{n,z}$ is in involution
with all elements of the set $\Pi_{n,z}$.

\begin{proof}
According to proposition \ref{cf1}, for $z=1$ one can take
$Z_{n,1}= P_n^2$ and form $L_{n,1}$ by collecting $2n-4$ suitable
elements of $P_n$. For $n=2$ we can only have $z=1$, $Z_{2,1}=
(P_{12})$ and $L_{2,1}=\emptyset$. In order to construct sets
$Z_{n,z}$ and $L_{n,z}$, with $n>2$ and $2\leq z \leq n-1$, we
shall proceed by induction on $n$.

Let us take $m>2$ and suppose that, for all $n=2, \dots, m-1$, we
have constructed sets $Z_{n,z}$ and $L_{n,z}$ of polynomial
functions of degree $\leq 2$ in $P_n$, for all possible $z=1,
\dots, n-1$, satisfying properties 1--4 specified above. Let us
split the set of indexes $N_m:=(1, \dots, m)$ into two arbitrary
nonempty disjoint subsets $I_1$ and $I_2$, such that $\sharp I_1=
n_1$, $\sharp I_2= n_2$, $n_1+ n_2=m$ and $N_m= I_1 \cup I_2$. For
$k=1,2$, consider the two sets of momenta $P^{(k)}\subset P_m$,
with $P^{(k)} := (P_{ij},\ i,j \in I_k,\ i<j)$ if $1< n_k\leq m-
1$, and $P^{(k)} := \emptyset$ if $n_k=1$. This means that the
elements of $P^{(k)}$ are the generators of the orthogonal
transformations of the subspace $\rn^{n_k}\subset \rn^m$ having
set of coordinates $(x_i,\ i\in I_k)$. We have obviously
$\{P^{(1)}, P^{(2)}\}= 0$. If $1< n_k\leq m- 1$, consider for any
$z_k$, with $1\leq z_k \leq n_k-1$, the sets
$Z_{n_k,z_k}(P^{(k)})$ and $L_{n_k,z_k}(P^{(k)})$, which are
obtained from the set of polynomials $Z_{n_k,z_k}$ and
$L_{n_k,z_k}$ by replacing the variables $P_{n_k}:=(P_{ij},\ 1\leq
i<j \leq n_k)$ with $P^{(k)}$. If instead $n_k=1$, take $z_k=0$,
$Z_{1,0}(P^{(k)}):= \emptyset$ and $L_{1,0}(P^{(k)}):= \emptyset$.
We thus have in all cases $\sharp Z_{n_k,z_k}(P^{(k)})= z_k$ and
$\sharp L_{n_k, z_k}(P^{(k)})=2(n_k- z_k-1)$. Finally, take
\begin{align*}
z&=z_1 + z_2+1\,, \\
Z_{m, z}&= \big(P_m^2, Z_{n_1,z_1}(P^{(1)}), Z_{n_2,z_2}(P^{(2)})
\big)\,, \\
L_{m, z}&= \big(L_{n_1,z_1}(P^{(1)}), L_{n_2,z_2}(P^{(2)})
\big)\,,
\end{align*}
where $P_m^2:= \sum_{1\leq i<j \leq m} P_{ij}^2$. It is easy to
see that the sets $Z_{m, z}$ and $L_{m, z}$ satisfy properties
1--4 above for $n=m$. If $n_1=1$, $z_1=0$, $n_2=m-1$, $z_2=1$, we
obtain $z=2$. With any other choice of $n_k$ and $z_k$, $k=1,2$,
$z$ can assume any value from 3 to $m-1$.
\end{proof}

Let the hamiltonian of a system have the form $H= f(p^2, r, P^2)$,
with $\partial f/\partial (p^2, r) \neq 0$ almost everywhere. From
proposition \ref{bamb} it follows that the sets of functions
$F_{n,z}:= (H, Z_{n,z}; L_{n,z})$ are integrable sets, with subset
of central elements $(H, Z_{n,z})$, for all $z=1, \dots, n-1$. We
have $\sharp F_{n,z}= 2n-k$ and $k=z+1$. Hence the number $k$ of
central elements can take all values from $k=2$ to $k=n$.

It has to be noted that all integrable sets with $k>2$, obtained
by means of proposition \ref{bamb}, depend on the choice of a
cartesian set of coordinates on $\rn^n$. This means that, to any
such set of coordinates, it corresponds in general an inequivalent
integrable set. If one performs a transformation of $SO(n)$ on
configuration space, then a given integrable set is transformed
into an equivalent one only if the transformation leaves invariant
all the central functions of the set.

As an example of application of proposition \ref{bamb}, in the two
following tables we show explicitly some integrable sets which are
obtained in the two cases $n=4$ and $n=5$ respectively. In these
tables we use the notation $P^2_{(123)} := P_{12}^2 +P_{13}^2+
P_{23}^2$ and $P^2_{(1234)} := P_{12}^2 +P_{13}^2 +P_{14}^2+
P_{23}^2 +P_{24}^2 +P_{34}^2$.

\begin{table}[h]
\centering
\begin{tabular}{l|c}
$F$ & $k$ \\
\hline
$(H, P^2; P_{13}, P_{14}, P_{23}, P_{24})$ & 2\\
$(H, P^2, P^2_{(123)}; P_{12}, P_{13})$ & 3 \\
$(H, P^2, P^2_{(123)}, P_{12})$ & 4 \\
$(H, P^2, P_{12}, P_{34})$ & 4
\end{tabular}
\caption{Integrable sets for $H= f(p^2, r, P^2)$ and $n=4$.}
\end{table}
\begin{table}[h]
\centering
\begin{tabular}{l|c}
$F$ & $k$ \\
\hline
$(H, P^2; P_{13}, P_{14}, P_{15}, P_{23}, P_{24}, P_{25})$ & 2\\
$(H, P^2, P^2_{(1234)}; P_{13}, P_{14}, P_{23}, P_{24})$ & 3 \\
$(H, P^2, P^2_{(1234)}, P^2_{(123)}; P_{13}, P_{23})$ & 4 \\
$(H, P^2, P^2_{(1234)}, P^2_{(123)}, P_{12})$ & 5 \\
$(H, P^2, P^2_{(1234)}, P_{12}, P_{34})$ & 5 \\
$(H, P^2, P^2_{(123)}, P_{45}; P_{12}, P_{13})$ & 4 \\
$(H, P^2, P^2_{(123)}, P_{45}, P_{12})$ & 5
\end{tabular}
\caption{Integrable sets for $H= f(p^2, r, P^2)$ and $n=5$.}
\end{table}

\subsection{Quantum particle in a central force field}\label{qcff}

The hamiltonian operator $\hat H$ of this system is obtained from
the hamiltonian function (\ref{claplace}) by standard quantization
(see definition %\ref{standard}),
in \cite{part1}), which here simply consists in the substitution
$p \to \partial/\partial x$ and the the replacement of the
multiplication of functions with the composition of corresponding
operators, in symbols: $\times \to \circ$. We thus obtain
\begin{equation}\label{laplace}
\hat H= \frac 12 \hat p^2 +U(r) \,,
\end{equation}
where
\begin{equation*}
r= \sqrt {x^2}\,, \qquad  x= (x_1, \dots, x_n)\,, \qquad \hat p =
(\hat p_1, \dots, \hat p_n)\,, \qquad \hat p_i = \frac
\partial {\partial x_i}\,.
\end{equation*}
The operator $\hat p^2$ is the Laplace operator in cartesian
coordinates. We can proceed as for the classical system, and on
the basis of the classical formulas we will obtain the
corresponding formulas where functions are converted into
operators and Poisson brackets into Lie brackets. In a similar way
we will also verify the quasi-independence of operators. Let us
fix cartesian coordinates in $\mathbb{R}^n_x$. Let $(x,p)$ be the
corresponding canonical coordinates on $\mathbb{R}^{2n}_{xp} =T^*
\mathbb{R}^n_x$. Let us consider the standard set of operators
$(x, \hat p)$, where $\hat p=\partial /\partial x$. This set is
obtained by canonical quantization from the coordinates $(x,p)$.
%see definition \ref{can}.
Therefore, identifying $(x,p)$ and $(x, \hat p)$ with the sets $B$
and ${\cal B}$ respectively, one can study the algebra of the
polynomial functions of these operators by applying propositions
% \ref{poiss}, \ref{appl1} and also remark \ref{sinx}.
3.1, 4.1, and also remark 4.4 of \cite{part2}. This
remark is useful in order to deal with arbitrary functions of $r$.

Let us consider the operators $\hat P_{ij}$, $1 \leq i <j \leq n$,
obtained by symmetric quantization from the classical momenta
$P_{ij} = x_i p_j -x_j p_i$. Taking into account the canonical
commutation relations
\begin{equation} \label{ccr}
[x_i,x_j]= 0\,, \qquad [\hat p_i,\hat p_j]= 0\,, \qquad [\hat p_i,
x_j] = \delta_{ij}
\end{equation}
for $i,j= 1, \dots, n$, where $\delta_{ij}$ is the Kr\"onecker
symbol, we have that $\hat P_{ij} = x_i \hat p_j -x_j \hat p_i$,
i.e., these operators coincide with the standard quantization of
momenta $P_{ij}$. From the quadratic dependence of $P_{ij}$ on
$(x,p)$, and from proposition
% \ref{poiss}
3.1 of \cite{part2} (case 1), it follows that the commutation
relations among the operators $\hat P_{ij}$ have the same form as
the Poisson brackets (\ref{ppoiss}) among the corresponding
classical functions:
\begin{equation} \label{pcomm}
[\hat P_{ij}, \hat P_{hk}] = -\delta_{ih} \hat P_{jk} -
\delta_{jk} \hat P_{ih} + \delta_{ik} \hat P_{jh} + \delta_{jh}
\hat P_{ik} \,.
\end{equation}
Similarly, from (\ref{rot}) and proposition %\ref{appl1}
4.1 of \cite{part2} it follows that
\begin{equation}\label{rot2}
[r^2,\hat P]=0\,,\qquad [\hat p^2,\hat P] =0\,.
\end{equation}
It is also easy to verify that $[U(r),\hat P]=0$ for any function
$U$, in accordance with the first of (\ref{rot2}) and with remark
% \ref{sinx}
4.4 of \cite{part2}. We thus conclude that $[\hat H, \hat P]=0$.
Note that $\{p^2, r^2\} = 4x \cdot p \neq 0$ and correspondingly
$[\hat p^2, r^2] = 2(x \cdot \hat p +\hat p \cdot x)= 4x \cdot
\hat p +2n\neq 0$.

Using proposition % \ref{wick}
2.5 of \cite{part2}, it is easy to check that the operator
$(P^2)^{\rm sym}$, obtained by symmetrization with respect to
$(x,\hat p)$ of the square length $P^2$ of momentum $P$, coincides
with the operator $\hat P^2= \sum_{i<j} \hat P_{ij}^2$ up to an
additive constant. We have in fact $\hat P^2= (P^2)^{\rm sym}+
n(n-1)/4$. Since the additive constant $n(n-1)/4$ is irrelevant
for Lie brackets, from classical relations
(\ref{rotp}) and from proposition %\ref{appl1}
4.1 (case b) of \cite{part2} we obtain
\begin{equation}\label{rotp2}
[\hat P^2, \hat P_{ij}]= [(P^2)^{\rm sym}, \hat P_{ij}] =0\,.
\end{equation}

The quasi-independence of the set $\hat{\Pi}= (\hat P^2, \hat L)$
of $2n-3$ operators follows from the functional independence of
the corresponding set of symbols $\Pi= (P^2, L)$, and from the
homogeneity of these functions with respect to $p$. In fact, the
latter property implies that these functions coincide with their
respective main parts. Let us consider the set $\hat F = (\hat H,
\hat P^2; \hat L)$, where $\hat H$ is given by (\ref{laplace}). If
$U(r)\equiv 0$, it is easy to check that the set of corresponding
symbols $F= (H, P^2; L)$ is functionally independent. Since these
functions are homogeneous with respect to $p$, these functions are
quasi-independent. Furthermore, since the main part of $H$ does
not depend on $U$, the property of quasi-independence is true for
arbitrary $U$. Therefore, the quantum system with hamiltonian
$\hat H$ is quasi-integrable with integrable set $\hat F = (\hat
H, \hat P^2; \hat L)$ and $k=2$ central operators, $\hat H$ and
$\hat P^2$.

It is possible to give a partial characterization of integrable
quantum systems which are invariant with respect to the action of
the group $G=SO(n)$ on $\mathbb{R}^{2n}_{xp}$, i.e., the systems
whose hamiltonian operator $\hat H$ commutes with vector operator
$\hat P$.

\begin{prop} \label{cf2}
If the hamiltonian operator $\hat H$ of a system has the form:
$\hat H=f(\hat p^2, \hat P^2, g_1(r), \dots, g_l(r))$, where the
function $f$ is an arbitrary noncommutative polynomial in the
$l+2$ variables $(\hat p^2, \hat P^2, g_1(r), \dots, g_l(r))$, and
$g_1(r), \dots, g_l(r)$ are arbitrary functions of $r$, then
$[\hat H,\hat P]=0$.

Viceversa, let $\hat H$ be an arbitrary operator of class ${\cal
O}$ on $K = \rn^n \setminus \{0\}$, such that $[\hat H, \hat
P]=0$. Then the symbol of its main part $M\hat H$ with respect to
linear momenta $\hat p$ (see % definition \ref{hg})
definitions in \cite{part1}) has the form $MH= g(p^2, r, P^2)$,
where $g$ is a homogeneous polynomial in the two variables $(p^2,
P^2)$, whose coefficients are arbitrary functions of $r$ defined
for all $r>0$. If the polynomial $g$ satisfies the condition
$\partial g/\partial (p^2, r) \neq 0$ almost everywhere, then the
systems is quasi-integrable with $k=2$, with central integrals
$\hat H$ and $P^2$. In certain cases it is possible to find an
additional integral, and to have quasi-integrability with $k=1$.
This is the case for example for the Newton potential, that is for
$\hat H= \hat p^2/2 -\alpha/r$, and for identical uncoupled
oscillators, that is for $\hat H= (\hat p^2 +r^2)/2$. The system
with hamiltonian function $\hat H$ which is a function only of the
square angular momentum, more exactly $\hat H = f(\hat P^2)$,
where $df\neq 0$ almost everywhere, is quasi-integrable with
$k=1$.

Let us present the integrable sets of operators $F= F(\hat H)$
which correspond to integrable quantum systems with these
hamiltonian operators. For $\hat H=f(\hat p^2, \hat P^2, g_1(r),
\dots, g_l(r))$, where the function $f$ is not locally
functionally dependent only on $\hat P^2$, one can take $F= (\hat
H, \hat P^2; \hat L)$, and therefore $k=2$. For $\hat H= \hat
p^2/2 -\alpha/r$ one can take $F= (\hat H; \hat P^2, \hat L, \hat
A_1)$, where $\hat A_1= \sum_{j=2}^n (\hat P_{1j}\hat p_j+ \hat
p_j \hat P_{1j})/2 - \alpha x_1/r$. For $\hat H= (\hat p^2
+r^2)/2$ one can take $F= (\hat H;\hat H_1, \hat H_2, \dots, \hat
H_{n-1},$ $\hat P_{12}, \hat P_{13}, \dots, \hat P_{1n})$, where
$\hat H_i = \frac 12 (\hat p_i^2 +x_i^2)$. For $\hat H = f(\hat
P^2)$ one can take $F= (\hat P^2;\hat p^2, r, \hat L)$. In all
cases except the first one, we have $k=1$.
\end{prop}
Since $x \cdot\hat p= (\hat p^2 r^2- r^2 \hat p^2)/4 - n/2$, this
proposition implies in particular that $[x \cdot\hat p, \hat
P]=0$. Note also the relation $\hat P^2= r^2 \hat p^2- (x
\cdot\hat p)^2 -(n-2)x \cdot\hat p$, which can for instance be
easily verified using proposition %\ref{prod}.
3.2 of \cite{part1}.

\begin{proof} Let the operator $\hat H$ of class
${\cal O}_K$ have the form $\hat H =f(\hat p^2, \hat P^2, g_1(r),$
$\dots, g_l(r))$, where the functions $f,g_1, \dots, g_l$ have the
properties specified in the proposition. Then the relation $[\hat
H, \hat P]=0$ follows from (\ref{rot2}) and (\ref{rotp2}).

Viceversa, let $\hat H$ be an operator of class ${\cal O}$ such
that $[\hat H,\hat P]=0$. Then lemma %\ref{14}
3.24 of \cite{part1} implies that $\{M H, P\}=0$, where $H$ and
$P$ are the symbols of $\hat H$ and $\hat P$ respectively. Using
lemma \ref{40} we thus obtain that $M H= g(p^2, r, P^2)$, where
$g$ is an arbitrary function of three variables. Furthermore,
since $p^2$ and $P^2$ are both homogeneous polynomials of order 2
in $p$, taking into account the definition of main part we obtain
that $g$ is a homogeneous polynomial in the two variables $(p^2,
P^2)$, whose coefficients are arbitrary functions of $r$ defined
for all $r>0$.

The proof of the remaining statements is similar to the proof of
the corresponding statements of proposition \ref{cf1}. Let us
consider, in particular, the hamiltonian $\hat H= \hat p^2/2
-\alpha/r$ of the quantum Kepler system \cite{Landau2}. Using
(\ref{rot2}),
together with proposition %\ref{pc1} and lemma \ref{lem1},
2.1 and lemma 2.2 of \cite{part2}, it is easy to verify that
$[\hat H, \hat A_i]=0$, where
\[
\hat A_i=\sum_{j=1}^n \hat P_{ij}\diamond \hat p_j- \alpha
\frac{x_i}r \,, \qquad i=1, \dots, n\,.
\]
We have used above the symbol $\diamond$ to denote symmetrized
products, as in \cite{part2}.
%according to definition \ref{symp}.
Only one component of $\hat A$ is quasi-independent of the set
$(\hat H, \hat P^2, \hat L)$. Note that we have in this case
\[
\hat A^2= 2\hat H\left[\hat P^2 -\left(\frac{n-1} 2\right)^2
\right] +\alpha^2\,.
\]

Finally, it is easy to check that the main parts of the operators
of the sets considered in the last part of the proposition are
functionally independent.
\end{proof}

The commutation relations between operators, from which the
integrability of the considered sets of operators has been
established, have been derived exploiting the quadratic dependence
of classical momenta $P_{ij}$ on the canonical variables $(x,p)$.
Let us now present an alternative proof of these relations, which
is only based on the linear dependence of these momenta on
impulses $p$. We shall consider the quantization of an arbitrary
vector field on configuration space $K$, more exactly, the
quantization of the hamiltonian function on $T^*K$ which
corresponds to the lifting of this field on $T^*K$. These
considerations are useful for the investigation of any linear
operator which is invariant with respect to the phase flows of
such vector fields on $K$, independently of the assumption that
these vector fields be linear.

Let $P=(P_1, \dots, P_l)$ be a set of functions on the symplectic
manifold $M=T^*K$, which are linear with respect to $p$:
\begin{equation} \label{linp}
P_i= v^0_i(x) + \langle p, v_i(x) \rangle \,,
\end{equation}
$P_i:T^*K \to \rn$. Such functions, in local coordinates $(x,p)$
induced by local coordinates $x$ on $K$, have the form $P_i=
P_i(x,p)= v^0_i(x)+ \sum_{k=1}^n v_i^k(x) p_k$. Let us suppose
that the linear combinations of these functions with constant
coefficients form a Lie algebra $\galg$ with Poisson brackets in
the role of commutators, and that the functions of the set $P$
form a basis of this algebra. Let us consider the set of operators
$\hat P=(\hat P_1, \dots, \hat P_l)$ obtained by standard
quantization from the set of functions $P$, i.e.,
\begin{equation} \label{diffp}
\hat P_i:= v^0_i(x)+ \sum_{k=1}^n v_i^k(x) \frac
\partial {\partial x_k}\,, \qquad i=1, \dots, l\,.
\end{equation}
Then the linear combinations of these operators form also a Lie
algebra with respect to the usual commutator of linear operators.
Moreover, let us consider the linear map, from the original Lie
algebra of functions on $M$ to the Lie algebra of operators, which
is defined by the correspondence of sets $P \to \hat P$, obtained
by standard quantization. This map is a isomorphism between these
two Lie algebras, i.e., the linear map preserves commutators.

This fact is an obvious consequence of the following more general
proposition. Let us consider the Lie algebra ${\cal V}= {\rm
Vect}_K (K\times \rn)$ of all vector fields defined on the direct
product $K\times \rn \ni (x,u)$, which do not depend on $u$. The
commutator in this algebra is the usual Lie bracket of vector
fields. In local coordinates $(x,u)$, $x= (x_1, \dots, x_n)$, a
vector field $V\in {\cal V}$ on the $(n+1)$-dimensional manifold
$K\times \rn$ has the form $\dot x= v(x)$, $\dot u=v^0 (x)$. Let
us consider also the algebra ${\cal F}$ of all functions $P$ on
the symplectic manifold $M=T^*K$ which are linear with respect to
the impulse $p$, i.e., functions of the form (\ref{linp}). Let us
also consider the Lie algebra ${\cal O}$ of all linear
nonhomogeneous differential operators on $K$, i.e., operators of
the form (\ref{diffp}).
\begin{prop} \label{isom}
There are canonical isomorphisms between these three Lie algebras
${\cal V}, {\cal F}, {\cal O}$. In local coordinates on $K$, the
coefficients $v_0(x), v_1(x), \dots, v_n(x)$ defining the elements
of these algebras are conserved under these isomorphisms.

Let $V$ and $W$ be two vector fields of class ${\cal V}= {\rm
Vect}_K (K\times \rn)$. Let us indicate the functions and
operators, associated with these fields, as $P_V, P_{W}$ and $\hat
P_{V}, \hat P_{W}$ correspondingly. Then the conservation of
commutators of these algebras under the considered isomorphisms
can be written in the form
\[
\{P_{V}, P_{W}\} = P_{[{V},{W}]} \,,\qquad [\hat P_{V}, \hat
P_{W}] =\hat P_{[{V}, {W}]} \,.
\]

The hamiltonian vector field $X_P$, defined by the hamiltonian
function $P\in {\cal F}$, can be lowered by natural projection
$\pi: T^*K \to K$. This means that $\pi_* (X_P(m)) \in T_x K$ does
not depend on the choice of the point $m \in \pi^{-1} (x)$, where
$X_P(m)$ is the vector field $X_P$ at point $m$, $\pi_*: TM \to
TK$ is the derivative of the map $\pi$, and $T_x K$ is the tangent
space to $K$ at point $x$. Suppose that $P=P_V$, where $V\in {\cal
V}$. Since the elements of $\cal V$ do not depend on $u$, the
vector field $V$ can be lowered onto configuration space $K$ via
the natural projection $K\times \rn_u \to K$. These two vector
fields, obtained by projection on $K$ from $P_V$ and $V$
respectively, are coincident.
\end{prop}
\begin{proof}
Both statements of this proposition, about the correspondence of
commutators of the three Lie algebras and about the coincidence of
the projections on $K$ of the two vector fields, can be easily
checked by direct computation in local coordinates. In general,
these statements are reformulations of simple well-known facts.
\end{proof}

Since classical momenta $P_{ij}(p,x)= x_i p_j -x_j p_i$ are
linearly dependent on classical impulses $p$, we can use the above
proposition to deduce the commutation relations (\ref{pcomm}) for
the operators $\hat P_{ij}$ from the corresponding classical
relations (\ref{ppoiss}).

Relations (\ref{ppoiss}) show that the set of momenta $P_{ij}$,
$1\leq i <j \leq n$, is a basis of the Lie algebra $\galg =so(n)$,
which is isomorphic to the Lie algebra of all skew-symmetric
matrices. A natural basis in this Lie algebra is formed by
matrices $D^{ij}$, $1\leq i <j \leq n$, defined by formula
(\ref{aij}). The correspondence $D^{ij} \mapsto P_{ij}$ is
extended to linear combinations of matrices $D^{ij}$ and functions
$P_{ij}$ as an isomorphism of Lie algebras. The basis $P=(P_{ij},
1\leq i <j \leq n)$ of Lie algebra $\galg$ induces a dual set of
coordinates on the co-algebra $\galg^*$. It is well-known that the
function $P^2:\galg^* \to \rn$, $P^2 := \sum_{1\leq i <j \leq n}
P_{ij}^2$, is an invariant of the co-adjoint representation of
$SO(n)$ on the co-algebra $\galg^*$, where $\galg = so(n)$. From
corollary %\ref{casim2},
4.4 of \cite{part2}, it follows that $[(P^2)^{\rm sym}_{\hat P},
\hat P] =0$, where $(P^2)^{\rm sym}_{\hat P}$ denotes the
symmetrization with respect to $\hat P$ of polynomial $P^2$. But
obviously $(P^2)^{\rm sym}_{\hat P}= \sum_{1\leq i <j \leq n} \hat
P_{ij}^2= \hat P^2$. We thus conclude that $[\hat P^2, \hat P]=
0$.

Let us prove now that $[r^2, \hat P]=[\hat p^2, \hat P]= 0$. Let
us consider the set $B= (1, x, p, P)$ of $l:= n(n-1)/2 +2n +1$
functions on $M^{2n}= T^*K$. From Poisson brackets relations
(\ref{ppoiss}), (\ref{p1}) and (\ref{p2}) it follows that this set
of functions is a basis in a $l$-dimensional Lie algebra. The
functions of set $B$ are linear nonhomogeneous functions of $p$.
Therefore proposition \ref{isom} implies that analogous
commutation relations hold for the operators $\bop =(1, x, \hat p,
\hat P)$ obtained from the functions of set $B$ by standard
quantization. We have in particular
\begin{align}
[x_i, \hat P_{jk}] &= \delta_{ij}x_k - \delta_{ik}x_j\,, \\
[\hat p_i, \hat P_{jk}] &= \delta_{ij} \hat p_k - \delta_{ik}\hat
p_j\,.
\end{align}
Then relations (\ref{rot2}) can be derived from (\ref{rot}) using
proposition %\ref{appl2},
4.2, case b, of \cite{part2}.

The following proposition is the quantum equivalent of proposition
\ref{bamb}.

\begin{prop}\label{bamb2}
For any $n\geq 2$ and for any $z=1, \dots, n-1$, it is possible to
construct in a recursive manner sets $\hat Z_{n,z}$ and $\hat
L_{n,z}$ of polynomial functions of degree $\leq 2$ in the
variables $\hat P_n:=(\hat P_{ij}, 1\leq i<j \leq n)$, with the
following properties:
\begin{enumerate}
\item $\hat Z_{n,z}$ contains $z$ elements,

\item $\hat L_{n,z}$ contains $2(n-z-1)$ elements, all of degree 1
in $\hat P_n$,

\item the set $\hat \Pi_{n,z}:= (\hat Z_{n,z}, \hat L_{n,z})$ is
quasi-independent,

\item $[\hat Z_{n,z}, \hat \Pi_{n,z}]=0$.
\end{enumerate}
\end{prop}
\begin{proof}
Let $\hat Z_{n,z}$ and $\hat L_{n,z}$ be the polynomials obtained
from the classical ones $Z_{n,z}$ and $L_{n,z}$ of proposition
\ref{bamb}, by simply replacing their arguments $P_n$ with the
quantized momenta $\hat P_n$. In this case symmetrization is
unnecessary, since in these polynomials all monomials of degree 2
are squares of elements of $\hat P_n$. Since all elements of the
set $\Pi_{n,z}= (Z_{n,z}, L_{n,z})$ are homogeneous polynomials in
$p$, they coincide with the symbol of the main parts of the
corresponding elements of $\hat \Pi_{n,z}$. Therefore the
quasi-independence of $\hat \Pi_{n,z}$ follows from the functional
independence of the classical set $\Pi_{n,z}$.

Point 4 can be proved just by repeating the proof of proposition
\ref{bamb}. When the set of indexes $N_m:=(1, \dots, m)$ is split
into two disjoint subsets $I_1$ and $I_2$, consider in fact the
two sets of operators $\hat P^{(k)}\subset \hat P_m$, $k=1,2$,
which are the standard quantization of the sets of momenta
$P^{(k)}$. From the isomorphism between the two Lie algebras
generated by the sets $P_n$ and $\hat P_n$ respectively, it
follows that $[\hat P^{(1)}, \hat P^{(2)}] = 0$. Moreover, since
$P^2_m$ is a Casimir function for the co-algebra $so(m)^*$, from
corollary %\ref{casim2}
4.4 of \cite{part2} it follows that $[\hat P^2_m, \hat
Z_{n_k,z_k}(\hat P^{(k)})] =[\hat P^2_m, \hat L_{n_k,z_k}(\hat
P^{(k)})]=0$ for $k=1,2$. Hence point 4 is obtained by induction
on $n$, as in the classical case.
\end{proof}

Let $\hat H$ be an arbitrary operator of class ${\cal O}$ on $K =
\rn^n \setminus \{0\}$, such that $[\hat H, \hat P]=0$. From
proposition \ref{cf2} it follows that $M H= g(p^2, r, P^2)$, where
the function $g$ is a homogeneous polynomial in the two variables
$(p^2, P^2)$, with coefficients dependent on $r$. If the
polynomial $g$ satisfies the condition $\partial g/\partial (p^2,
r) \neq 0$ almost everywhere, from proposition \ref{bamb2} it
follows that the sets of operators $\hat F_{n,z}:= (\hat H, \hat
Z_{n,z}; \hat L_{n,z})$ are quasi-integrable sets, with subset of
central elements $(\hat H, \hat Z_{n,z})$, for all $z=1, \dots,
n-1$. We have $\sharp \hat F_{n,z}= 2n-k$ and $k=z+1$. Hence the
number $k$ of central elements can take all values from $k=2$ to
$k=n$.

\section{Free rotation of an $n$-dimensional rigid body} \label{meth2}

The configuration space $K$ of this system is the group $SO(n)$ of
orthogonal transformations of the $n$-dimensional euclidean space
$K= SO(n)$. Its phase space is $M=T^*K= T^*SO(n)$. We shall
present a detailed analysis of the integrability of the classical
system, which distinguishes itself from other already existing
investigations \cite{mish,ratiu,adler}, in the fact that we apply
here the concept of noncommutative integrability
%in the sense of definition \ref{clint}.
(see definition 3.16 of \cite{part1}), and analyze the dependence
of the number $k$ of central integrals on the properties of the
set of generalized moments of inertia, more precisely on the
presence in this set of subsets consisting of moments equal to
each other. The results will then be applied to the investigation
of the integrability of the corresponding quantum system. We begin
however with a preliminary subsection on some properties of the
group $SO(n)$, most of which are more or less well-known, but
which shall here be derived in a form and with a notation
convenient for our present purposes.

\subsection{Properties of left- and right-invariant vector fields
on $SO(n)$}\label{fr}

Any arbitrary Lie group $G$ acts on itself by left and right
shifts: for each element $g\in G$ these are defined by the
diffeomorphisms
\begin{align*}
&L_g:G \to G, \qquad L_gh=gh \,, \\
&R_g:G \to G, \qquad R_gh=hg \,.
\end{align*}
Let us denote by $T_g G$ the tangent space to $G$ at point $g$.
Then the Lie algebra $\galg$ associated with $G$ can be identified
with the tangent space $\galg= T_eG$ at the neutral element $e$ of
the group $G$. Let $(L_g)_*$ and $(R_g)_*$ respectively denote the
derivatives of the maps $L_g$ and $R_g$ at $e$. Each element $a\in
\galg$ defines two vector fields $V_a^L$ and $V_a^R$ on $G$. At
any point $g\in G$ these vector fields are respectively defined as
$V^L_a(g):= (L_g)_*a \in T_g G$ and $V^R_a(g):= (R_g)_*a \in T_g
G$. We have
\begin{align*}
(L_h)_* V^L_a (g) &= (L_h)_*(L_g)_*a = (L_{hg})_* a = V^L_a (hg)
\,, \\
(R_h)_* V^R_a (g) &= (R_h)_*(R_g)_*a = (R_{gh})_* a = V^R_a (gh)
\,.
\end{align*}
These relations mean that the vector field $V^L_a$ (respectively
$V^R_a$) is invariant with respect to the action of Lie group $G$
on itself by left (respectively right) shifts.
\begin{defn}
For the above reason, the fields $V^L_a$ and $V^R_a$ are
respectively called {\it left-invariant} and {\it right-invariant
vector field on $G$} associated with the element $a\in \galg$.
\end{defn}
The Lie brackets of left-invariant vector fields respect the
structure of the Lie algebra $\galg$: $[V^L_a, V^L_b] =
V^L_{[a,b]}$, where $a,b \in \galg$ and $[a,b]$ is their
commutator in $\galg$. An analogous result, although with a
reversed sign, is true for right-invariant vector fields: $[V^R_a,
V^R_b] = -V^R_{[a,b]}$. We have also $[V^L_a, V^R_b]=0$.
% Il segno delle Lie brackets dei vector fields, così come quello delle
% parentesi di Poisson, è qui invertito rispetto all'Arnold!

The Lie group $G=SO(n)$ can be identified with the group of
orthogonal matrices of size $n\times n$, and its associated Lie
algebra $\galg = so(n)$ with the Lie algebra of skew-symmetric
matrices of the same size. Then the vector field $V_A^L$
(respectively $V_A^R$) corresponding to the skew-symmetric matrix
$A$ is defined by the system of differential equations $\dot X=
XA$ (respectively $\dot X= AX$), where $X$ is an orthogonal
matrix. The commutator of the Lie algebra $\galg= so(n)$ is given
by the usual commutator of matrices. We can therefore rewrite the
above formulas for the Lie brackets between left- and
right-invariant vector fields as
\begin{equation} \label{commv}
[V^L_A,V^L_B]= V^L_{[A,B]}\,, \qquad [V^R_A,V^R_B]=
-V^R_{[A,B]}\,, \qquad [V^L_A,V^R_B]=0\,.
\end{equation}
We can also consider the classical impulses $P^L_A, P^R_A: T^*M
\to \rn$, with $M= SO(n)$, which are associated according to
proposition \ref{isom} with the vector fields $V^L_A$ and $V^R_A$,
for any $A\in \galg$:
\begin{equation} \label{pla}
P^L_A (m)= \langle p, V^L_{A}(X) \rangle\,, \qquad P^R_A (m)=
\langle p, V^R_{A}(X) \rangle\,,
\end{equation}
where $p=m \in T^*_X G$. These impulses form a Lie algebra with
respect to Poisson brackets, which is isomorphic to the Lie
algebra of the corresponding vector fields. Therefore, from
relations (\ref{commv}) we derive
\begin{equation} \label{commp}
\{P^L_A,P^L_B\}= P^L_{[A,B]}\,, \qquad \{P^R_A,P^R_B\}=
-P^R_{[A,B]}\,, \qquad \{P^L_A,P^R_B\}=0\,.
\end{equation}

Let us consider the vector fields $V_{ij}^L$, $V_{ij}^R$, which
are associated with the matrices $D^{ij}$ defined by formula
(\ref{aij}). We recall that the set $(D^{ij}, 1\leq i< j\leq n)$
forms a basis of the Lie algebra $\galg=so(n)= T_e G$. This
implies that $V_{ij}^L(g)$ and $V_{ij}^R(g)$ are two bases of
linear space $T_g G$, for any $g\in G$. From (\ref{acomm}) and
(\ref{commv}) we obtain that the Lie brackets between these vector
fields have the form
\begin{align}
[V_{ij}^L, V_{hk}^L] &= -\delta_{ih} V^L_{jk} - \delta_{jk}
V^L_{ih} +
\delta_{ik} V^L_{jh} + \delta_{jh} V^L_{ik} \,, \label{vlvl}\\
[V_{ij}^R, V_{hk}^R] &= \delta_{ih} V^R_{jk} + \delta_{jk}
V^R_{ih} -
\delta_{ik} V^R_{jh} - \delta_{jh} V^R_{ik} \,, \\
[V_{ij}^L, V_{hk}^R] &= 0 \,. \label{vlvr}
\end{align}

These vector fields are associated with the classical impulses
%$P^L_{ij}= P^L_{ij}(m)$ and $P^R_{ij}= P^R_{ij}(m)$:
\begin{equation} \label{iplpr}
P^L_{ij}(m)= \langle p, V^L_{ij}(X) \rangle\,, \qquad P^R_{ij}(m)=
\langle p, V^R_{ij}(X) \rangle\,.
\end{equation}
According to (\ref{commp}), the Poisson brackets between these
functions have the same form as the Lie brackets
(\ref{vlvl})--(\ref{vlvr}) between the corresponding vector
fields:
\begin{align}
\{P_{ij}^L, P_{hk}^L\} &= -\delta_{ih} P^L_{jk} - \delta_{jk}
P^L_{ih} +
\delta_{ik} P^L_{jh} + \delta_{jh} P^L_{ik} \,, \label{plpl}\\
\{P_{ij}^R, P_{hk}^R\} &= \delta_{ih} P^R_{jk} + \delta_{jk}
P^R_{ih} -
\delta_{ik} P^R_{jh} - \delta_{jh} P^R_{ik} \,, \label{prpr}\\
\{P_{ij}^L, P_{hk}^R\} &= 0 \,. \label{prpl}
\end{align}
%with $1\leq i< j\leq n$ and $1\leq h< k\leq n$.

In the following we shall often indicate with $P^L=P^L(m)$ and
$P^R=P^R(m)$ the two skew-symmetric matrices having components
$P^L_{ij}(m)$ and $P^R_{ij}(m)$ respectively, with $1\leq i, j\leq
n$. Since for any $A\in \galg$ we have $A= \sum_{i<j}A_{ij}
D^{ij}$, from (\ref{pla}) we obtain $P^L_A= \sum_{i<j}
A_{ij}P^L_{ij}= -\frac 12 \sum_{i,j} P^L_{ij}A_{ji} = -\frac 12
\Tr(P^L A)$. Using the first of relations (\ref{commp}) we then
obtain
\begin{equation} \label{plu}
\{P^L_A,P^L_{hk}\}= -\frac 12 \Tr(P^L[A,D^{hk}])= -\frac 12
\Tr(D^{hk}[P^L,A]) =[P^L,A]_{hk}\,,
\end{equation}
where the last member represents the element at row $h$ and column
$k$ of the commutator of the two matrices $P^L$ and $A$. In a
similar way we obtain
\begin{equation} \label{pru}
\{P^R_A,P^R_{hk}\}= -[P^R,A]_{hk}\,.
\end{equation}
\begin{rem} \label{typ}
Since the $N$ vectors $V^L_{ij}(X)$, $1\leq i<j \leq n$, form a
linear basis of $T_X G$ for any $X\in G=SO(n)$, we see from
formula (\ref{iplpr}) that the correspondence $p \mapsto P^L(m)$
is an isomorphism between the linear spaces $T^*_X G$ and $so(n)$.
It follows that, for a generic $m\in M$, the matrix $P^L(m)$ is a
``typical'' $n \times n$ skew-symmetric matrix. The same fact is
obviously true also for $P^R(m)$.
\end{rem}

Let us represent the generic element $g$ of the group $G=SO(n)$ as
an orthogonal $n\times n$ real matrix $X$. If we indicate as
$\tilde X$ the transposed of the matrix $X$, so that $\tilde
X_{\beta\gamma} := X_{\gamma\beta}$, the orthogonality of $X$
implies $X^{-1} = \tilde X$. For any skew-symmetric matrix $A=
-\tilde A$, at any point $X\in SO(n)$ the tangent vector $\dot X=
AX$ can also be written as $\dot X= XB$, where $B=X^{-1}AX$ is
also a skew-symmetric matrix. Using this fact, it is easy to see
that the basic left- and right-invariant vector fields introduced
above are connected to each other by the relations
\begin{equation*}
V_{ij}^R(X)= \sum_{h,k=1}^n X_{ih} X_{jk} V_{hk}^L(X) \,,
\end{equation*}
or in matrix notation
\begin{equation}
V^R(X)= X V^L(X) \tilde X\,. \label{mpv}
\end{equation}
Similarly, we have for the corresponding impulses
\begin{equation} \label{py}
P_{ij}^R(m)= \sum_{h,k=1}^n X_{ih} X_{jk} P_{hk}^L(m)\,,
\end{equation}
or equivalently
\begin{equation}
P^R(m)= X P^L(m) \tilde X\,, \label{mpy}
\end{equation}
where $m\in T^*_X G$.

It is sometimes useful to indicate with $P^L$ and $P^R$ also the
two sets of $N=n(n-1)/2$ functions on $T^*G$ defined by formulas
(\ref{iplpr}):
\begin{equation} \label{splpr}
P^L = (P^L_{ij}, 1\leq i<j \leq n)\,, \qquad P^R = (P^R_{ij},
1\leq i<j \leq n)\,.
\end{equation}
In this way the elements of the sets $P^L$ and $P^R$, evaluated at
a point $m\in M$, just coincide with the independent elements of
the two skew-symmetric matrices $P^L(m)$ and $P^R(m)$
respectively, which were introduced before formula (\ref{plu}). In
the following, it should be clear from the context whether $P^L$
(or $P^R$) indicates a skew-symmetric matrix or the corresponding
set of $N$ independent elements.

Let $F=(F_1, \dots, F_r)$ be a set of $r$ real functions $F_i:
M\to \rn$, $i=1, \dots, r$, on a manifold $M$. As usual, we denote
as $\rank F$ at $m\in M$ the dimension of the image $F_* (T_m M)$
of tangent space $T_m M$ with respect to the derivative $F_*$ of
the map $F:M \to \rn^r$ defined by the set $F$. Let $dF=(dF_1,
\dots, dF_r)$ denote the set of the differentials of the elements
of $F$ at the point $m$. It is well known that $\rank F= \dim{\rm
Span}\,dF$, where ${\rm Span}\,dF$ denotes the linear subspace of
$T^*_m M$ spanned by the elements of $dF$.

Let $M=M^{2N}$ be a symplectic manifold. Let us denote with $\Pi$
the antisymmetric bilinear functional $\Pi:T^*_m M \times T^*_m
M\to \rn$ such that the Poisson bracket of any two functions $f,h:
M\to \rn$ is given at $m$ by the relation $\{f,h\}= \Pi(df, dh)$.
% see formula (\ref{pbra}).
For any linear subspace $L \subseteq T^*_m M$, let us denote with
$L^\angle$ the subspace skew-orthogonal to $L$ with respect to
$\Pi$, i.e., $L^\angle= \{u\in T^*_m M: \Pi(u,w)=0\,\forall \,w\in
L\}$. Since $\Pi$ is nondegenerate, we have $\dim L+ \dim
L^\angle= 2N$. From these facts we deduce the following
\begin{lem} \label{rskew}
Let $F$ be a set of functions on a symplectic manifold $M^{2N}$.
Then
\begin{equation}
\dim \left({\rm Span}\,dF\right)^\angle= 2N- \rank F\,.
\end{equation}
\end{lem}

\begin{prop}\label{r0}
Let us consider the two set of functions $P^L$ and $P^R$ defined
by formulas (\ref{splpr}) and (\ref{iplpr}). On all $M=T^*G$,
where $G=SO(n)$, we have
\begin{equation}
\rank P^L = \rank P^R = N\,.  \label{rank}
\end{equation}
In addition, the two subspaces ${\rm Span}\,dP^L$ and ${\rm
Span}\,dP^R$ of $T^*_m M$ are related to each other by the
equalities
\begin{equation} \label{pskew}
({\rm Span}\,dP^L)^\angle = {\rm Span}\,dP^R\,, \qquad {\rm
(Span}\,dP^R)^\angle ={\rm Span}\,dP^L\,.
\end{equation}
\end{prop}
\begin{proof}
Let us fix a system of coordinates $x$ in a neighborhood of $g\in
G$, and let $(x, p)$ be the local coordinates induced from $x$ on
$T^*G$. According to (\ref{iplpr}), we have $V^L_{ij} =\partial
P_{ij}^L/
\partial p$. Since at any point $g$ of $G$ the $N$ vector fields
$V^L_{ij}$, $1\leq i<j \leq n$, are linearly independent, we have
that the differentials $d_p P^L_{ij}\in T_g G$ with respect to
variables $p$ of the $N$ functions $P^L_{ij}$ are linearly
independent. This implies in particular the independence of their
differentials $dP^L_{ij}\in T_m^*M$ with respect to variables
$(x,p)$, so that $\rank P^L= N$. We obtain in a similar way that
$\rank P^R= N$. Equalities (\ref{rank}) are thus proved.

Formula (\ref{prpl}) implies ${\rm Span}\,dP^R \subseteq ({\rm
Span}\,dP^L)^\angle$. Furthermore, formula (\ref{rank}) is
equivalent to $\dim {\rm Span}\,dP^L= \dim {\rm Span}\,dP^R =N$.
Therefore, applying lemma \ref{rskew} to the set $P^L$ we obtain
$\dim ({\rm Span}\,dP^L)^\angle = N =\dim {\rm Span}\,dP^R$. These
facts imply the former of equalities (\ref{pskew}). The latter is
obtained in a similar way.
\end{proof}

Let $B$ denote the set $B:= (P^L, P^R)$ of $2N=n(n-1)$ functions
on $T^*G$.

\begin{cor}
On all $M=T^*G$ we have
\begin{equation} \label{spandba}
({\rm Span}\,dB)^\angle= {\rm Span}\,dP^L \cap {\rm Span}\,dP^R\,.
\end{equation}
\end{cor}
\begin{proof}
Recalling equalities (\ref{pskew}), from the relation $B = P^L
\cup P^R$ we obtain $({\rm Span}\,dB)^\angle =({\rm
Span}\,dP^R)^\angle \cap ({\rm Span}\,dP^L)^\angle ={\rm
Span}\,dP^L\cap {\rm Span}\,dP^R$.
\end{proof}

In order to establish the linear dimension of subspace
(\ref{spandba}), it is first necessary to recall a basic property
of skew-symmetric matrices.

\begin{lem} \label{skew}
Let us consider a skew-symmetric operator in an $n$-dimensional
euclidean space $L$. Then there is a cartesian system of
coordinates in which the operator is defined by a matrix which has
the ``normal block-diagonal'' form:
\begin{equation} \label{block}
\bar A= \left(\begin{array}{cc|cc|c|cc|c} 0 & \alpha_1 & 0 &0 & \cdots & 0&0&0 \\
-\alpha_1 & 0 & 0 & 0 &\cdots & 0&0&0\\
\hline 0 & 0& 0 & \alpha_2 & \cdots & 0&0&0\\
0& 0& -\alpha_2 & 0& \cdots & 0&0&0\\
\hline \vdots & \vdots & \vdots &\vdots & \ddots & \vdots&\vdots&\vdots\\
\hline 0 & 0 & 0 & 0 & \cdots &0&\alpha_s& 0 \\
0 & 0 & 0 & 0 & \cdots &-\alpha_s& 0&0 \\
\hline 0 & 0 & 0 & 0 & \cdots &0&0& 0
\end{array} \right) \,.
\end{equation}
where $s=\left[\frac n2 \right]$ is the integer part of $\frac
n2$, $\alpha_k \geq0$ for $k=1, \dots, s$, and the last vanishing
row and column are present only for odd $n$. Let us suppose that
all the eigenvalues $\pm i\alpha_1, \dots, \pm i\alpha_s, (0)$ of
$\bar A$ (where $i= \sqrt{-1}$) are pairwise different. (This
means that $\alpha_k>0$, $\alpha_h\neq \alpha_k\,\forall\, h,k=1,
\dots, s,\ h\neq k$). Then a skew-symmetric matrix $B$ commutes
with $\bar A$, i.e., $[\bar A,B]=0$, if and only if $B$ has the
same block-diagonal form, without any condition on its
eigenvalues.

An equivalent intrinsic formulation is the following. For any
skew-symmetric operator $A$ with pairwise different eigenvalues on
the euclidean space $L$, there exists a unique decomposition $L =
L_1^2 \oplus L^2_2 \oplus \dots \oplus L^2_s \oplus L^1_0$ of $L$
into $s$ 2-dimensional invariant subspaces $L^2_i$ and (for odd
$n$) a 1-dimensional subspace $L^1_0$. Any skew-symmetric operator
$B$ commutes with $A$ if and only if all subspaces $L^2_i$ and
$L^1_0$ of this decomposition are invariant under the action of
$B$.
\end{lem}

We shall denote with $\gtyp\subset \galg = so(n)$ the set of all
skew-symmetric matrices whose eigenvalues are pairwise different.
Obviously, ``almost all'' elements of $\galg$ also belongs to
$\gtyp$. More exactly, $\galg\setminus\gtyp$ is a closed subset of
$\galg$ having vanishing measure. It follows that, if some
statement is true in $\gtyp$, then it is true (at least) almost
everywhere in $\galg$.

Let us fix an element $a$ of the Lie algebra $\galg$, and consider
the linear operator ${\rm ad}_a: \galg \to \galg$ in this algebra,
${\rm ad}_a: b \mapsto [a,b]$. The kernel of ${\rm ad}_a$ is the
subalgebra of all elements of $\galg$ which commute with $a$,
i.e., ${\rm Ker}\, {\rm ad}_a =\{b \in \galg : [a,b]=0\}$.

\begin{cor} \label{n/2}
If $a\in \gtyp$, then
\begin{equation}\label{kera}
\dim{\rm Ker}\, {\rm ad}_a= \ipn2\,.
\end{equation}
Let ${\rm ad}_a(\galg)\subset \galg$ denote the image of $\galg$
with respect to the operator ${\rm ad}_a$. Then
\begin{equation}\label{adg}
\dim {\rm ad}_a(\galg)= N-\ipn2\,.
\end{equation}
Let $a$ be represented by a matrix $\bar A$ having the normal
block-diagonal form (\ref{block}), with $\alpha_k>0$,
$\alpha_h\neq \alpha_k\,\forall\, h,k=1, \dots, s,\ h\neq k$. Then
${\rm ad}_a(\galg)$ is the linear space of the matrices $B\in
\galg$ such that $B_{12}=B_{34}=\dots =B_{2s-1,2s}=0$, with
$s=[n/2]$, i.e.,
\begin{equation}\label{ada}
{\rm ad}_{\bar A}(\galg)= \{B\in \galg:B_{2i-1,2i}=0 \ \forall\,
i=1, \dots, [n/2]\}\,.
\end{equation}
\end{cor}
\begin{proof}
Equality (\ref{kera}) follows from lemma \ref{skew}, and from the
observation that the linear space of the matrices of the form
(\ref{block}) has dimension $s= \ipn2$. Equality (\ref{adg})
follows directly from (\ref{kera}), since $\dim \galg= N= n(n-
1)/2$. A direct computation shows that, for any $C\in \galg$, one
has $B_{2i-1,2i}=0 \ \forall\, i=1, \dots, \ipn2$, where $B:=[\bar
A, C]$. Therefore the linear space on the right-hand side of
(\ref{ada}) contains ${\rm ad}_{\bar A}(\galg)$. Furthermore, it
is obvious that the right-hand side of (\ref{ada}) has linear
dimension $N-\ipn2$. Taking into account (\ref{adg}), these facts
imply (\ref{ada}).
\end{proof}

For a given point $m\in M=T^*SO(n)$, let us consider the two
skew-symmetric matrices $P^L(m)$ and $P^R(m)$ whose elements are
defined by formula (\ref{iplpr}). Let us denote with $\mtyp\subset
M$ the set of all points $m\in M$ such that $P^L(m)\in \gtyp$. We
see from formula (\ref{mpy}) that $P^L(m)\in \gtyp$ if and only if
$P^R(m)\in \gtyp$. We also know (see remark \ref{typ}) that, for a
generic $m\in M$, the matrix $P^L(m)$ is a typical $n \times n$
skew-symmetric matrix. Hence, the eigenvalues of $P^L(m)$ are
pairwise different for almost all $m\in M$. It means that almost
all elements of $M$ belong to $\mtyp$, or equivalently that
$M\setminus\mtyp$ is a closed subset of $M$ having vanishing
measure.

Let $\sigma= \sigma(m)$ be the linear dimension of the kernel of
the linear operator ${\rm ad}_{P^L}$:
\begin{equation}\label{sigma}
\sigma:=\dim {\rm Ker}\, {\rm ad}_{P^L}\,.
\end{equation}
In other words, $\sigma$ is the linear dimension of the subalgebra
of all matrices $A\in so(n)$ such that $[P^L,A]=0$.

\begin{lem}\label{sigman}
For all $m\in M$, where $M=T^*SO(n)$, we have
\begin{equation}\label{sigmar}
\dim {\rm Ker}\, {\rm ad}_{P^R}= \dim {\rm Ker}\, {\rm ad}_{P^L}=
\sigma \,.
\end{equation}
In addition, for all $m\in \mtyp$ (hence, for almost all $m\in M$)
we have
\begin{equation}\label{sign2}
\sigma= \ipn2\,.
\end{equation}
\end{lem}
\begin{proof}
According to formula (\ref{mpy}), $\forall\,A\in \galg= so(n)$ we
have
\[
[P^L,A]=[X^{-1}P^R X,A]=X^{-1}[P^R,{\rm Ad}_X (A)]X\,,
\]
where ${\rm Ad}_X (A)=XAX^{-1}$ is the image of $A$ with respect
to adjoint action ${\rm Ad}_X: \galg \to \galg$ of $X$ on $\galg$.
Hence, $A\in {\rm Ker}\, {\rm ad}_{P^L}$ if and only if ${\rm
Ad}_X (A)\in {\rm Ker}\, {\rm ad}_{P^R}$. This implies
(\ref{sigmar}), since the adjoint action ${\rm Ad}_X$ is an
algebra automorphism of $\galg$. Equality (\ref{sign2}) follows
from formula (\ref{kera}) of corollary \ref{n/2}.
\end{proof}

\begin{prop}\label{br}
Let us consider the set $B=(P^L, P^R)$ of $2N$ functions defined
by formula (\ref{iplpr}). At any point $m\in M=T^*G$, where
$G=SO(n)$, we have
\begin{align}
(\sspan dB)^\angle&= \Big\{a\in T^*_m M:a=\sum_{i<j} A_{ij}
dP^L_{ij}\,, \, \tilde A=-A\,, \,
[P^L,A]=0\Big\} \nonumber\\
&= \Big\{a\in T^*_m M:a=\sum_{i<j} A_{ij} dP^R_{ij}\,, \, \tilde
A=-A\,,\, [P^R,A]=0\Big\} \label{kpr}
\end{align}
and
\begin{equation} \label{rankb}
\rank B = 2N-\sigma \,,
\end{equation}
where $\sigma$ is defined by formula (\ref{sigma}).
\end{prop}

\begin{proof}
Equalities (\ref{pskew}) and (\ref{spandba}) imply that
\[
({\rm Span}\,dB)^\angle = {\rm Span}\,dP^L \cap ({\rm Span}\,
dP^L)^\angle \,.
\]
Therefore $({\rm Span}\,dB)^\angle$ is the set of all those
elements $a \in {\rm Span}\,dP^L$ such that $\Pi(a,dP^L_{hk})=0\
\forall\, h,k$ with $1\leq h<k \leq n$. According to (\ref{rank}),
$dP^L$ is a set of linearly independent elements of $T^*_m G$.
Hence any covector $a\in {\rm Span}\,dP^L$ can be expressed as $a=
\sum_{i<j} A_{ij} dP^L_{ij}$, where $A_{ij}$ are elements of a
univocally determined skew-symmetric matrix $A$. We have
\begin{equation}\label{pipl}
\Pi(a,dP^L_{hk})= \sum_{i<j} A_{ij} \Pi(dP^L_{ij}, dP^L_{hk})
=\{P^L_A, P^L_{hk}\}= [P^L,A]_{hk}\,,
\end{equation}
where for the last equality use has been made of formula
(\ref{plu}). Therefore $a\in ({\rm Span}\, dP^L)^\angle$ if and
only if $[P^L,A]=0$. This implies the first equality of formula
(\ref{kpr}); the second one can be obtained in a similar way.

Formula (\ref{kpr}) shows that $({\rm Span}\,dB)^\angle$ is in
one-to-one correspondence with the subalgebra of matrices $A\in
so(n)$ such that $[P^L,A]=0$. Hence
\begin{equation} \label{spandbo}
\dim ({\rm Span}\,dB)^\angle= \dim \ker {\rm ad}_{P_L}= \sigma\,,
\end{equation}
so that (\ref{rankb}) is obtained by applying lemma \ref{rskew} to
the set $B$.
\end{proof}

From formulas (\ref{rankb}) and (\ref{sign2}) we immediately
obtain the following:
\begin{cor}\label{br1}
At all points of $\mtyp$ (hence, almost everywhere on $M=T^*G$) we
have
\begin{equation} \label{rank2}
\rank B = 2N-\left[\frac n2\right]\,.
\end{equation}
\end{cor}

The previous corollary, together with proposition \ref{strind2},
implies that there exist at most $\ipn2$ functionally independent
functions on $M$ which are in involutions with the whole set $B$.
In order to explicitly construct a set of $\ipn2$ such functions,
it is useful to extend to Lie algebras the notion of Casimir
function which has already been introduced for Lie co-algebras in
% definition \ref{defcas2}.
\cite{part2}.

\begin{defn} \label{defcas}
A {\it Casimir function} on a Lie algebra $\galg$ is an invariant
of the adjoint action of the local Lie group $G$, i.e., a function
which is constant on the orbits of this action. In other words, a
Casimir function $K:\galg \to \rn$ is a first integral which is
common to all differential equations in $\galg$ of type $\dot b =
[a, b]$, where $b\in \galg$ and $a$ is a fixed element of algebra
$\galg$.

Let us consider the case $\galg=so(n)$, i.e., $\galg$ is the
algebra of skew-symmetric matrices. Fix $A\in \galg$ and consider
the characteristic polynomial $\det(A-\lambda E)$, where $E$ is
the identity matrix. The relation $\tilde A=-A$ implies
$\det(A+\lambda E)= \det(\tilde A+\lambda E)= (-1)^n (A-\lambda
E)$, so that the characteristic polynomial has the form
$\det(\lambda E-A)= \lambda^n+ C_{1}(A) \lambda^{n-2}+ C_{2}(A)
\lambda^{n-4} + \cdots$. The coefficients $C_1(A), \dots, C_s(A)$
of this polynomial, where $s= \left[\frac n2 \right]$, are clearly
polynomial functions of the $N=n(n-1)/2$ independent elements
$A_{ij}$, $1\leq i<j \leq n$, of the skew-symmetric matrix $A$.
Such elements are the coefficients of $A$ with respect to the
basis $D^{ij}$ of $\galg$ given by formula (\ref{aij}). We will
call $C=(C_1, \dots, C_s)$ the {\it standard set of Casimir
functions} on $so(n)$.
\end{defn}

\begin{lem} \label{stcas}
The functions of the standard set $C$ of Casimir functions on
$\galg= so(n)$ are really Casimir functions in the sense of
definition \ref{defcas}. Moreover, the functions of set $C$ form a
basis in the space of all Casimir functions in the following
``functional'' sense. The differentials of the functions of set
$C$ are linearly independent at all points of $\gtyp$: therefore,
they are linearly independent at almost all points of $\galg$.
Moreover, every Casimir function $K$ on $\galg$ can be locally
expressed as a function of the elements of this set: $K=K(C)$.
\end{lem}
\begin{proof}
These are actually well-known facts, but we give the proof for
completeness. The adjoint action of the group $G=SO(n)$ on the Lie
algebra $\galg =so(n)$ takes in the matrix representation the form
$A \mapsto {\rm Ad}_X(A)=XAX^{-1}$, where $X\in G$, $A\in \galg$.
Therefore $\det(A -\lambda E)= \det({\rm Ad}_X(A) -\lambda E)$,
i.e., the polynomial $\det(A -\lambda E)$ in the variable
$\lambda$ is invariant under the adjoint action of $G$. This means
that the coefficients of this polynomial, i.e., the functions of
set $C$, are also invariant, and are thus Casimir functions in the
sense of definition \ref{defcas}.

Let us now examine the functional independence of the $s$
functions of set $C$, where $s=\ipn2$, at a given point $\bar B\in
\galg$. To this end, we shall evaluate the rank of the $s\times N$
matrix
\[
J= J(\bar B):= \left(\left.\frac{\partial C_h(B)} {\partial
B_{ij}}\right|_{B= \bar B} \,, \ 1\leq h \leq s\,, \ 1\leq i<j
\leq n\right)\,,
\]
whose rows are labelled by the index $h$ and the columns by the
double index $ij$. In the above formula, we have denoted with
$\partial C_h(B)/ \partial B_{ij}$ the partial derivatives of the
Casimir function $C_h$ with respect to the independent elements
$B_{ij}$ of the skew symmetric matrix $B$. We know from lemma
\ref{skew} that there exists $X= X(\bar B)\in G$ such that ${\rm
Ad}_X(\bar B) =\bar A$, where $\bar A$ is a skew-symmetric matrix
in the normal form (\ref{block}), such that $\alpha_k\geq 0$ for
$k=1, \dots, s$, and $\pm i\alpha_1, \dots, \pm i\alpha_s, (0)$
are the eigenvalues of $\bar B$. Since $C$ is a set of Casimir
functions, we have that $C(\bar B)=C(\bar A)$. Taking into account
that ${\rm Ad}_X$ is an invertible linear operator in $\galg$, it
follows that $\rank J(\bar B) =\rank J(\bar A)$.

Let $P(B, \lambda):= \det(\lambda E-B)$ denote the characteristic
polynomial for the matrix $B\in \galg$. For any $\lambda\in \rn$
it is easy to see that
\begin{equation} \label{offd}
\left. \frac{\partial P(B, \lambda)} {\partial B_{ij}}
\right|_{B=\bar A} =0 \quad \forall \,i,j=1, \dots, n \mbox{ such
that }j>i+1\,.
\end{equation}
Recalling the definition of the set $C$, equality (\ref{offd})
implies that all the columns of matrix $J(\bar A)$, corresponding
to indexes $ij$ with $j>i+1$, are zero. We can thus write $\rank
J(\bar A)= \rank \tilde J(\bar A)$, where we have introduced the
$s\times s$ matrix
\begin{equation} \label{kkb}
\tilde J= \tilde J(\bar A):= \left(\left.\frac{\partial C_h(B)}
{\partial B_{2i-1, 2i}} \right|_{B=\bar A}\,, \ 1\leq h \leq s\,,
\ 1\leq i \leq s\right)\,.
\end{equation}

It is easy to see that, for a matrix $\bar A$ of the form
(\ref{block}), the characteristic polynomial has the form
$\det(\lambda E -\bar A)= (\lambda) \prod_{i=1}^s(\lambda^2+
\alpha_i^2) =\lambda^n +\sum_i \alpha_i^2 \lambda^{n-2}+
\sum_{i<j} \alpha_j^2 \alpha_j^2\lambda^{n-4}+ \dots$. Therefore
the standard set of Casimir functions can be written as $C(\bar
A)= D(\beta(\alpha))$, where $\alpha=(\alpha_1, \dots, \alpha_s)$,
$\beta(\alpha)= (\beta_1(\alpha), \dots, \beta_s(\alpha)):=
(\alpha_1^2, \dots, \alpha_s^2)$, $D(\beta)= (D_1(\beta), \dots,
D_s(\beta))$, and
\[
D_h(\beta):= \sum_{i_1 <i_2 <\cdots <i_h} \beta_{i_1} \beta_{i_2}
\cdots \beta_{i_h}  \quad \mbox{for} \quad h=1, \dots, s\,.
\]
Note that in particular $D_1(\beta)= \sum_{i=1}^s \beta_i$ and
$D_s(\beta) = \beta_1 \beta_2 \cdots \beta_s$. Since $\bar
A_{2i-1, 2i}= \alpha_i$, we have
\[
\left.\frac{\partial C_h(B)} {\partial B_{2i-1, 2i}} \right|_{B
=\bar A} =2\alpha_i \left.\frac{\partial D_h(\beta)} {\partial
\beta_i}\right|_{\beta=\beta(\alpha)}\,.
\]
Therefore, recalling (\ref{kkb}), we have $\det \tilde J(\bar A)=
2^s \alpha_1 \alpha_2 \cdots \alpha_s K_s(\beta(\alpha))$, where
$K_s(\beta)$ is the determinant of the $s\times s$ jacobian matrix
$\frac{\partial D}{\partial \beta}(\beta)$. It is easy to see that
$K_s(\beta)$ is a symmetric polynomial of order $s(s-1)/2$ in the
variables $\beta$, which vanishes whenever $\beta_i= \beta_j$ for
some $i\neq j$. Therefore we can write $K_s(\beta) = c_s
\prod_{i<j} (\beta_i - \beta_j)$. In order to evaluate the
constant coefficient $c_s$, we note that, when $\beta_s=0$, one
has $\partial D_s/
\partial \beta_i =0$ for $i=1, \dots, s-1$, and $\partial D_s/
\partial \beta_s =\beta_1 \beta_2 \cdots \beta_{s-1}$. It easily
follows that $K_s(\beta)= \beta_1 \beta_2 \cdots
\beta_{s-1}K_{s-1}(\beta')$ for $\beta_s=0$, where $\beta'= (\beta
_1, \dots, \beta_{s-1})$. Using this fact, it can be easily shown,
by induction with respect to $s$, that $c_s=1\,\forall\, s\in
\nn$. We thus conclude that
\begin{equation} \label{tildej}
\det \tilde J(\bar A)= 2^s \alpha_1 \alpha_2 \cdots \alpha_s
\prod_{i<j} (\alpha^2_i - \alpha^2_j)\,.
\end{equation}

When all eigenvalues of matrix $\bar B$ are pairwise different, we
have $\alpha_k>0$, $\alpha_h\neq \alpha_k\,\forall\, h,k=1, \dots,
s,\ h\neq k$. In such cases we see therefore from (\ref{tildej})
that $\det \tilde J(\bar A)\neq 0$, so that $\rank J(\bar B)=
\rank \tilde J(\bar A)=s$. This means that the differentials of
the $s$ Casimir functions of set $C$ are linearly independent at
$\bar B$.

In order to complete the proof of the lemma, we note that,
according to corollary \ref{n/2}, at almost each point $B\in
\galg$ the set of vectors $[A,B]$, where $A$ varies on all
$\galg$, forms a subspace of codimension $s=\left[\frac n2\right]$
in $\galg$. Since the differential of any Casimir function must be
zero when acting on this subspace, we deduce that the
differentials of any $r$ Casimir functions, where $r\geq s$, are
linearly dependent at any point $B\in \galg$. Therefore the
differential of any Casimir function $K$ at any point of algebra
$so(n)$ is a linear combination of the differentials of the
functions of the standard set, i.e., we have locally $K=K(C)$.
\end{proof}

Note that the set of all Casimir functions on $\galg$ defines the
orbits ${\cal O}$ of the adjoint representation of the
corresponding local Lie group $G$ by their common level surfaces:
$C^{-1}(c)= {\cal O}$, where $C^{-1}(c)$ is the pre-image of a
point $c \in \rn^s$, and $s=\left[\frac n2\right]$. More exactly,
this is true in the domain of regularity of the map $C:\galg \to
\rn^s$, defined by the set $C$. From this it is easy to obtain
that the typical orbit ${\cal O}$ has codimension $s$ in $\galg$,
i.e., $\dim {\cal O} = \dim \galg -s$.

For $\galg= so(n)$, we can associate with every element $b\in
\galg^*$ a unique skew symmetric matrix $B$ such that $(A,b)=
\frac 12 \Tr (\tilde A B)= \sum_{i<j} A_{ij} B_{ij}\ \forall \,
A\in \galg$. Clearly, the elements $B_{ij}$, $1\leq i<j\leq n$, of
matrix $B$ are just the coefficients of $b$ with respect to the
dual basis of the basis $D^{ij}$ of $so(n)$ defined by formula
(\ref{aij}). According to this one-to-one correspondence, in the
following we will often identify $\galg^*$ with the co-algebra of
skew-symmetric matrices, and we shall denote with $\gtyp^*$ the
set of all elements of $\galg^*$ whose eigenvalues are pairwise
different. Obviously, almost all elements of $\galg^*$ also
belongs to $\gtyp^*$. It is easy to see that, in the matrix
representation, the co-adjoint action of the group $G=SO(n)$ on an
element $A\in \galg^*$ has the same form as the adjoint action on
the corresponding element of $\galg$, i.e., $A\mapsto X A X^{-1}$,
with $X\in G$. It follows that the standard set of Casimir
functions on the Lie algebra $so(n)$, introduced in definition
\ref{defcas}, also represents a standard set of Casimir functions
on the space of skew-symmetric matrices, when the latter is
identified with the co-algebra $so(n)^*$. This can be formally
stated as follows.

\begin{defn} %\label{defstcas}
Let $\galg^*$ be the dual co-algebra of $\galg=so(n)$. Fix a
skew-symmetric matrix $A\in \galg^*$ and consider the
characteristic polynomial $\det(\lambda E-A)= \lambda^n+ C_{1}(A)
\lambda^{n-2}+ C_{2}(A) \lambda^{n-4} + \cdots$. The coefficients
$C_1(A), \dots,$ $C_s(A)$ of this polynomial, where $s= \left[
\frac n2 \right]$ is the integer part of $\frac n2$, are
polynomial functions of the set $\alpha$ of the matrix elements of
$A$. We will call $C=(C_1, \dots, C_s)$ the {\it standard set of
Casimir functions} on $so(n)^*$.
\end{defn}

\begin{lem} \label{stcas2}
The functions of the standard set $C$ of Casimir functions on the
coalgebra $\galg^*$, where $\galg=so(n)$, are really Casimir
functions in the sense of the definition %\ref{defcas2}.
given in \cite{part2}. Moreover, the functions of set $C$ form a
basis in the space of all Casimir functions in the following
``functional'' sense. The differentials of the functions of set
$C$ are linearly independent at all points of $\gtyp^*$.
Therefore, they are linearly independent at almost each point of
$\galg^*$. Moreover, every Casimir function $K$ on $\galg^*$ can
be locally expressed as a function of the elements of this set:
$K=K(C)$.
\end{lem}

\begin{prop}\label{coin}
Let $C$ be the standard set of Casimir functions on the Lie
algebra $so(n)$. Then $C(P^L)= C(P^R)$ on all symplectic manifold
$M=T^*G$.
\end{prop}
\begin{proof}
From relation (\ref{mpy}), taking into account that $\tilde X=
X^{-1}\,\forall\, X\in SO(n)$, we obtain that $\det(\lambda
E-P^R)=\det(\lambda E-P^L)$ on all $M$ for any $\lambda \in \rn$.
The thesis then follows from definition \ref{stcas} of the set
$C$.
\end{proof}
In the following we shall denote with $\bar C$ the set $\bar
C:=C(P^L)= C(P^R)$ of $s=\ipn2$ real functions on $M$.

\begin{prop}\label{r1}
On all $M=T^*G$, where $G=SO(n)$, we have
\begin{equation} \label{cb}
\{\bar C, B\}_M =0  \,.
\end{equation}
Moreover, on all $\mtyp$ (hence, almost everywhere on $M$) we have
\begin{align}
\rank \bar C &= \left[ \frac n2 \right] \,, \label{n2}\\
%\rank B &= 2N-\left[\frac n2\right]\,, \label{rank2} \\
{\rm Span}\,d\bar C &= ({\rm Span}\,dB)^\angle \,. \label{dbskew}
\end{align}
\end{prop}

\begin{proof}
Since $\bar C= C(P^L)$, formula (\ref{prpl}) implies $\{\bar C,
P^R\} = 0$. On the other hand, according to proposition \ref{coin}
we also have $\bar C= C(P^R)$, so that using again (\ref{prpl}) we
obtain $\{\bar C, P^L\} = 0$. Equality (\ref{cb}) is thus proved.

Since both $\galg$ and $\galg^*$ have been identified with the set
of $n\times n$ skew-symmetric matrices, we have $\mtyp= \{m\in M :
P^L(m) \in \gtyp^*\}$. Lemma \ref{stcas} implies that $\rank C=
\left[\frac n2 \right]$ at all points of $\gtyp$. Moreover,
according to (\ref{rank}), $P^L:M \to \galg$ is a regular map.
Therefore, it follows from the definition $\bar C= C\circ P^L$
that equality (\ref{n2}) holds at all points $m\in \mtyp$.

Formula (\ref{cb}) implies that ${\rm Span}\,d\bar C \subseteq
({\rm Span}\,dB)^\angle$. Moreover, from (\ref{n2}),
(\ref{spandbo}), and (\ref{sign2}), we obtain that $\dim {\rm
Span}\,d\bar C= \left[ \frac n2 \right]= \dim ({\rm
Span}\,dB)^\angle$ at all points of $\mtyp$. These relations imply
(\ref{dbskew}).
\end{proof}

Taking into account (\ref{pskew}) and (\ref{spandba}), equality
(\ref{dbskew}) also implies
\begin{equation} \label{spanc}
{\rm Span}\,d\bar C = {\rm Span}\,dP^L \cap {\rm Span}\, dP^R
={\rm Span}\,dP^L \cap ({\rm Span}\, dP^L)^\angle \,.
\end{equation}

\begin{rem}\label{r1rem}
For any function $f:M\to \rn$, the condition $\{f,B\}=0$ is
equivalent to $df \in ({\rm Span}\,dB)^\angle$. According to
formula (\ref{dbskew}), the latter condition implies that locally
$f=g(\bar C)$, where $g:\rn^{[n/2]} \to \rn$ is an arbitrary
function. Recalling the definition of the set $\bar C$, this means
that $f=g(C(P^L)) = c(P^L)$, where $C$ is the standard set of
Casimir functions, and therefore $c:=g\circ C$ is also a Casimir
function. Taking into account proposition \ref{coin}, we also have
$f=c(P^R)$. We conclude that the functions on $M$ which are in
involution with the whole set $B$ are all and only the Casimir
functions of the set $P^L$ (or equivalently $P^R$).
\end{rem}

\subsection{Free rotation of a classical rigid body}

Let us consider the system with hamiltonian function of the form
\begin{equation} \label{hamrb}
H= \frac 1 2 \langle P^L, JP^L \rangle \,,
\end{equation}
where $J$ is an operator in the $N$-dimensional linear space
$\rn^N$ of impulses $P^L$, with $N=n(n-1)/2$. The map $P^L: T^*G
\to \rn^N= \galg$ is linear with respect to impulses $p$, see
formula (\ref{iplpr}). Therefore the function $H=H(P^L(X))$,
$H:T^*G \to \rn$, at any point $X\in G$ defines a quadratic form
on cotangent space $T^*_X G$ to $G$ at $X$, and this quadratic
form $H|_X :=H|_{T^*_X G}$ is invariant with respect to the action
of $G$ on itself by left shifts. This fact is the motivation of
the following definition.
\begin{defn}
We say that a system with hamiltonian function (\ref{hamrb})
describes the {\it free motion of a point on a Lie group $G$
provided with left-invariant metric}.
\end{defn}

Note that if the restriction $H|_{X=e}$ of $H$ on $T^*_e G$ is
positively defined, then $H$ defines on $G$ a left-invariant
riemannian metric, and the system with hamiltonian $H$ defines a
geodesic flow on cotangent bundle $T^*G$ which corresponds to this
metric. ($H$ can be considered as the kinetic energy $T$ of the
point.) From a mathematical point of view, the positivity of the
2-form $H|_{X=e}$, which is transferred on all $G$ by left shifts,
is not important (see \cite{Arnold}).

\begin{defn}
Let the operator $J$ in formula (\ref{hamrb}) be diagonal, and let
its element be such that
\begin{equation} \label{hamlam}
H= H^{\lambda}= \frac 12 \sum_{i<j} \frac
{\left(P^L_{ij}\right)^2}{\lambda_i +\lambda_j} \,,
\end{equation}
with $\lambda_i >0$ for $i=1, \dots, n$. In this particular case
we also say that the system with hamiltonian function
(\ref{hamrb}) describes the rotation of a {\it rigid body in
$n$-dimensional space, with generalized moments of inertia
$\lambda_1, \dots, \lambda_n$}. In the following of this section
we will always consider systems of this type.
\end{defn}

\begin{lem} \label{hpr}
If $H=H(P^L)$, then $\{H,P^R\}=0$.
\end{lem}
This result follows from relation (\ref{prpl}).

\begin{lem} \label{hpl}
Let $H=H^{\lambda}$ have the form (\ref{hamlam}), and suppose that
$\lambda_i=\lambda_j$ for some $i,j$. Then $\{H,P_{ij}^L\}=0$.
\end{lem}

\begin{proof}
From (\ref{plpl}) and (\ref{hamlam}) we get for all $i,j$, $1\leq
i<j \leq n$,
\[
\{H,P_{ij}^L\}= (\lambda_i -\lambda_j) \sum_{k=1}^n \frac
{P^L_{ik} P^L_{kj}} {(\lambda_i +\lambda_k) (\lambda_k
+\lambda_j)} \,.
\]
The right-hand side obviously vanishes when $\lambda_i=
\lambda_j$. Note that the above formula is equivalent to the
well-known Euler equations for the free rotation of a rigid body.
\end{proof}

Denoting $\lambda=(\lambda_1, \dots, \lambda_n)$, let us group
together the elements of the set $\lambda$ which are equal to each
other. More precisely, we suppose that there exist $u$ distinct
positive real numbers $\mu_1, \dots, \mu_u$, with $1\leq u \leq n$
and $\mu_h \neq \mu_k$ for $h\neq k$, such that
\begin{equation}\label{mu}
\begin{split}
\lambda_1 = \lambda_2= \dots=\lambda_{p_1}&= \mu_1 \,,\\
\lambda_{p_1+1}= \lambda_{p_1+2} = \dots = \lambda_{p_2} &= \mu_2 \,,\\
& \ \,\vdots \\
\lambda_{p_{u-1}+1} =\lambda_{p_{u-1}+2} =\dots =\lambda_{p_u}&=
\mu_u \,, \end{split}
\end{equation}
with $p_u=n$. Let us consider the set $q= q(\lambda) =(q_1, \dots,
q_u)$, where $q_1=p_1$, $q_2=p_2-p_1$, $\dots$, $q_n=p_n-p_{n-1}$.
The set of all possible $\lambda$ can be divided into classes
$l(q)$ each characterized by a given set $q= (q_1, \dots, q_u)$,
with $q_j \in \nn$ for $j=1, \dots, u$, and $\sum_{j=1}^u q_j =n$.
The order of the different generalized moments of inertia $\mu_h$
is not important, so one may always arrange them in such a way
that $1\leq q_1\leq q_2 \leq\dots \leq q_u$.

For a given $\lambda$, consider the decomposition of $\rn^n$ which
corresponds to this $\lambda$:
\begin{equation}\label{dec}
\rn^n = L_1 \oplus L_2\oplus\dots \oplus L_u \,,
\end{equation}
where $\dim L_j = q_j$, $\lambda\in l(q)$, $q=(q_1, \dots, q_u)$.
Let us consider the subalgebra $\galg(\lambda) \subseteq
\galg=so(n)$ of all skew-symmetric matrices $A$ such that the
subspaces $L_j$ are invariant with respect to the action of the
operators defined by these matrices for all $j=1, \dots, q$:
\[
\galg^\lambda= \{A\in \galg : A(L_j) \subseteq L_j\ \forall\, j=1,
\dots, u\} \,.
\]
Let us consider a cartesian basis in $\rn^n$ corresponding to the
decomposition (\ref{dec}), and the set of skew-symmetric matrices
which represent the elements of Lie algebra $\galg$ in this basis.
A basis of $\galg$ is then given by the set of $N$ matrices
$D^{ij}$ of the form (\ref{aij}), with $1\leq i<j\leq n$. This
basis is orthonormal with respect to the euclidean structure
induced in $\galg$ by the bilinear form $(A, B):= {\rm Tr}\,
(\tilde A B)/2$, for $A, B \in \galg$. Let $I^\lambda$ denote the
set of the pairs of indexes which correspond to equal generalized
moments of inertia. Recalling equalities (\ref{mu}), we have
$I^\lambda=(I^1, \dots, I^u)$, where
\begin{equation} \label{ilam}
I^k := \{(i,j)\in \nn \times\nn : p_{k-1}< i<j\leq p_k\}\quad
\text{for} \quad k=1, \dots, u\,.
\end{equation}
Then the set $D^\lambda:= \{D^{ij}: (i,j)\in I^\lambda\}$ of
matrices of the form (\ref{aij}) forms a basis of $\galg^\lambda$.
Let us consider the set $P^{L\lambda} := \{P^{L}_{ij} :(i,j)\in
I^\lambda\} \subseteq P^L$ of impulses of the form (\ref{iplpr}).
The following proposition is an obvious consequence of lemmas
\ref{hpr} and \ref{hpl}.
\begin{prop}
If $H=H^{\lambda}$ has the form (\ref{hamlam}), then
$\{H,B^\lambda\}=0$, where $B^\lambda :=(P^{L\lambda},P^R)$.
\end{prop}

In other words, $B^\lambda$ contains the elements of $B$ which are
in involution with the hamiltonian $H^\lambda$. We are now going
to investigate whether it is possible to construct an integrable
set using suitable functions of the set $B^\lambda$ and possibly
$H^\lambda$. We will find that this depends on the properties of
the set $\lambda$ or, more exactly, of the set of integers
$q(\lambda)$.

Let us consider a set of functions $F=(F_1, \dots, F_p)$ on the
$2N$-dimensional symplectic manifold $M^{2N}$. We suppose that
$\rank F$ is the same almost everywhere. Let $k=k(F)$ be the
maximal number of functionally independent functions, defined on
all $M^{2N}$, which are functions of $F$ and are in involution
with all functions of set $F$. More exactly, there exists a set
$Z=(Z_1, \dots, Z_k)$, $Z_j=Z_j(F)$, such that these functions are
functionally independent almost everywhere on $M$ and $\{Z,F\}=0$
on $M$, and $k$ is the maximum integer for which such a set $Z$
exists.

\begin{defn}\label{kr}
We will say that the number $k$ is the {\it centrality} of the set
$F$, while $r=r(F):= 2N- \rank F - k(F)$ is the {\it defect of
integrability} of the set $F$.
\end{defn}

\begin{lem} \label{kf}
For any set $F$ (such that $\rank F$ is the same almost
everywhere) we have $k(F)\leq \dim W$, where $W:= \sspan dF \cap
(\sspan dF)^\angle$ at a typical point $m\in M$. Moreover,
$r(F)\geq 0$.
\end{lem}
\begin{proof}
For any function $Z=Z(F)$, such that $\{Z,F\}=0$ at $m$, we must
obviously have $z\in W$. From this it easily follows that
$k(F)\leq \dim W$. The inequality $r(F)\geq 0$ follows instead
from proposition \ref{strind2}.
\end{proof}

\begin{prop} \label{kb}
The centrality and defect of integrability of the set of functions
$B=(P^L, P^R)$ are respectively
\begin{align}
k(B)&=\ipn2 \,, \label{kbn2}\\
r(B)&=0 \,. \label{rb}
\end{align}
\end{prop}
\begin{proof}
Equality (\ref{kbn2}) follows from proposition \ref{r1} and remark
\ref{r1rem}. Then equality (\ref{rb}) follows from (\ref{rank2}).
\end{proof}

Let $\pop^\lambda: \galg \to \galg^\lambda$ denote the projector
onto the subalgebra $\galg^\lambda$ with respect to the euclidean
structure in $\galg$ introduced above. For any linear operator
$L:\galg \to \galg$, we denote by $L|_{\galg^\lambda}$ its
restriction to the subalgebra $\galg^\lambda$. For any $m\in M$ we
can then consider the operators
\begin{align*}
\pop^\lambda \circ {\rm ad}_{P^L} &: \galg \to \galg^\lambda\,,\\
\pop^\lambda \circ {\rm ad}_{P^L}|_{\galg^\lambda} &:
\galg^\lambda \to
\galg^\lambda\,, \\
{\rm ad}_{P^L}|_{\galg^\lambda} &: \galg^\lambda \to \galg\,,
\end{align*}
where $\circ$ denotes the composition of linear operators, and
$P^L=P^L(m)$. We define the three integers $\sigma_1^\lambda$,
$\sigma_2^\lambda$, and $\sigma_3^\lambda$, dependent on the point
$m$, as the linear dimension of the kernels of the three above
operators:
\begin{align}
\sigma_1^\lambda &:= \dim {\rm Ker} \left(\pop^\lambda \circ {\rm
ad}_{P^L}\right)
\,, \label{sigma1}\\
\sigma_2^\lambda &:= \dim {\rm Ker} \left(\pop^\lambda \circ
{\rm ad}_{P^L}|_{\galg^\lambda}\right) \,, \label{sigma2}\\
\sigma_3^\lambda &:= \dim {\rm Ker} \left ({\rm
ad}_{P^L}|_{\galg^\lambda} \right) \,. \label{sigma3}
\end{align}
From these definitions and from (\ref{sigma}) it is immediate to
see that $\sigma_1^\lambda \geq \sigma\geq \sigma_3^\lambda$ and
$\sigma_1^\lambda \geq \sigma_2^\lambda\geq \sigma_3^\lambda$.

\begin{prop} \label{defect0}
The rank of the set $B^\lambda=(P^{L\lambda},P^R)$ at any point
$m\in M$ is given by
\begin{equation}
\rank B^\lambda = 2N- \sigma_1^\lambda\,. \label{rblam}
\end{equation}
In addition we have
\begin{equation}\label{dimw2}
\dim W= \sigma+ \sigma_2^\lambda- \sigma_3^\lambda\,,
\end{equation}
where
\begin{equation}\label{wb}
W:=\sspan dB^\lambda \cap (\sspan dB^\lambda)^\angle\,.
\end{equation}
\end{prop}

\begin{proof}
Since $B^\lambda = P^{L\lambda}\cup P^R$, we have
\begin{align} \label{spandbl}
(\sspan dB^\lambda)^\angle &= (\sspan dP^R)^\angle \cap (\sspan
dP^{L\lambda})^\angle \nonumber \\
&= \sspan dP^L\cap (\sspan dP^{L\lambda})^\angle\,,
\end{align}
where the last equality follows from (\ref{pskew}). Therefore
$({\rm Span}\,dB^\lambda)^\angle$ is the set of all those elements
$a \in {\rm Span}\,dP^L$ such that $\Pi(a,dP^L_{hk})=0\ \forall
\,(h,k)\in I^\lambda$. Any covector $a\in {\rm Span}\,dP^L$ can be
expressed as $a= \sum_{i<j} A_{ij} dP^L_{ij}$, where $A_{ij}$ are
elements of a univocally determined skew-symmetric matrix $A$.
Hence $a=dP^L_A$ and
\begin{equation} \label{piapl}
\Pi(a,dP^L_{hk})= \{P^L_A, P^L_{hk}\}= [P^L,A]_{hk}\,,
\end{equation}
where for the last equality use has been made of formula
(\ref{plu}). The matrix elements on the right-hand side of the
above formula, for $(h,k)\in I^\lambda$, are just the components,
with respect to the basis $D^\lambda$, of the projection
$\pop^\lambda([P^L,A])$ of the element $[P^L,A] \in \galg$ onto
the subalgebra $\galg^\lambda$. We see therefore that $({\rm
Span}\,dB^\lambda)^\angle$ is in one-to-one correspondence with
the linear space of matrices $A\in so(n)$ such that
$\pop^\lambda([P^L,A])=0$. Hence
\begin{equation} \label{spandblo}
\dim ({\rm Span}\,dB^\lambda)^\angle= \dim {\rm Ker}
\left(\pop^\lambda \circ {\rm ad}_{P^L}\right)=
\sigma_1^\lambda\,,
\end{equation}
so that (\ref{rblam}) is obtained by applying lemma \ref{rskew} to
the set $B^\lambda$.

From (\ref{spandbl}) and (\ref{wb}) it follows that
\begin{equation}%\label{wb}
W= \sspan dB^\lambda \cap (\sspan dP^R)^\angle \cap (\sspan
dP^{L\lambda})^\angle\,.
\end{equation}
Let us then consider an arbitrary element $w\in W$. The condition
$w\in \sspan dB^\lambda$ implies that there exist $z_1\in \sspan
dP^{R}$ and $z_2\in \sspan dP^{L\lambda}$ such that
\[
w= z_1 +z_2\,.
\]
Since $P^{L\lambda}\subseteq P^{L}$, from (\ref{pskew}) it follows
that
\begin{equation} \label{zlr}
\Pi(z_1, dP^{L\lambda})= \Pi(z_2, dP^R)= 0\,.
\end{equation}
Hence the condition $w\in (\sspan dP^R)^\angle$ implies $0=
\Pi(w,dP^R) =\Pi(z_1,dP^R)$. Therefore, recalling (\ref{pskew})
and (\ref{spandba}), we have
\[
z_1\in \sspan dP^R \cap (\sspan dP^R)^\angle= (\sspan
dB)^\angle\,.
\]
Using proposition \ref{br} we thus conclude that $z_1=
\sum_{i<j}C_{ij}dP^L_{ij}$, where $C$ is a skew-symmetric matrix
such that $[P^L,C] =0$, i.e., $C\in {\rm Ker}\, {\rm ad}_{P^L}$.
Finally, the condition $w\in (\sspan dP^{L\lambda})^\angle$, using
again (\ref{zlr}), implies $0= \Pi(w,dP^{L\lambda})
=\Pi(z_2,dP^{L\lambda})$. Therefore
\[
z_2\in \sspan dP^{L\lambda} \cap (\sspan dP^{L\lambda})^\angle\,.
\]
Using formula (\ref{piapl}), we see that a covector $a\in {\rm
Span}\,dP^L$ belongs to $\sspan dP^{L\lambda} \cap (\sspan
dP^{L\lambda})^\angle$ if and only if $a= \sum_{i<j} A_{ij}
dP^L_{ij}$, with $A\in \galg^\lambda$ and $\pop^\lambda
([P^L,A])=0$. These two conditions on $A$ can be simultaneously
expressed as $A\in S^\lambda$, where $S^\lambda:= {\rm Ker}
\left(\pop^\lambda \circ {\rm ad}_{P^L}|_{\galg^\lambda}\right)$.

We have thus shown that $w=\sum_{i<j} (C_{ij}+ A_{ij}) dP^L_{ij}$,
where $C\in S:={\rm Ker}\, {\rm ad}_{P^L}$ and $A\in S^\lambda$.
It follows that
\begin{equation}\label{w}
W =\Big\{w\in T^*_m M:w=\sum_{i<j} K_{ij} dP^L_{ij}\,,\, K\in
\sspan(S, S^\lambda)\Big\}\,.
\end{equation}
Since covectors $dP^L_{ij}$, $1\leq i<j\leq n$, are linearly
independent (see proposition \ref{r0}), according to (\ref{w})
there exists a linear isomorphism between $W$ and $\sspan(S,
S^\lambda)$. Therefore
\begin{equation}\label{dimw}
\dim W= \dim \sspan(S, S^\lambda) =\dim S +\dim S^\lambda -\dim(S
\cap S^\lambda)\,.
\end{equation}
Recalling (\ref{sigma}) and (\ref{sigma2}), we have $\dim S=
\sigma$ and $\dim S^\lambda= \sigma_2^\lambda$. Furthermore, it is
easy to see that $S \cap S^\lambda= \ker {\rm ad}_{P^L} \cap
\galg^\lambda = \ker {\rm ad}_{P^L}|_{\galg^\lambda}$, so that
recalling (\ref{sigma3}) we can write $\dim(S \cap S^\lambda)=
\sigma_3^\lambda$. Hence (\ref{dimw}) is equivalent to
(\ref{dimw2}).
\end{proof}

\begin{cor}\label{defect}
At a typical point $m\in M$, then the centrality $k(B^\lambda)$
and the defect of integrability $r(B^\lambda)$ (see definition
\ref{kr}) of the set $B^\lambda$ satisfy the inequalities
\begin{align}
k(B^\lambda) &\leq\sigma+ \sigma_2^\lambda-
\sigma_3^\lambda\,, \label{kblam}\\
r(B^\lambda) &\geq \sigma_1^\lambda+ \sigma_3^\lambda- \sigma -
\sigma_2^\lambda \label{dblam}\,.
\end{align}
\end{cor}
\begin{proof}
Relation (\ref{kblam}) follows from lemma \ref{kf} and formula
(\ref{dimw2}). Then relation (\ref{dblam}) follows from
(\ref{rblam}) and (\ref{kblam}).
\end{proof}

Note that formula (\ref{wb}) implies $\dim W \leq \dim (\sspan
dB^\lambda)^\angle$. Therefore using (\ref{spandblo}) and
(\ref{dimw2}) we obtain the inequality $\sigma_1^\lambda \geq
\sigma+ \sigma_2^\lambda- \sigma_3^\lambda$.

Let us consider again the set $q= (q_1, \dots, q_u)$ introduced
after formula (\ref{mu}). If $u=1$, i.e., $q=(n)$, $0<\lambda_1=
\lambda_2 = \dots =\lambda_n$, then $\galg^\lambda =\galg$, so
that $\sigma_1^\lambda =\sigma_2^\lambda= \sigma_3^\lambda=
\sigma$ and $B^\lambda=B$. We recall that, according to lemma
\ref{sigman}, we have $\sigma= \left[\frac n2\right]$ for a
typical $m\in M$. Therefore, in this particular case formulas
(\ref{rblam}) and (\ref{kblam})--(\ref{dblam}) agree with the
results $\rank B= 2N- \ipn2$, $k(B)=\ipn2$, and $r(B)=0$, see
formula (\ref{rank2}) and proposition \ref{kb}.

\begin{lem} \label{q1b}
Suppose that $u>1$, i.e., there exist at least two generalized
moments of inertia which are different from each other. Then at a
typical point $m\in M$ we have
\begin{align}
\sigma_1^\lambda &= \sum_{i<j} q_i q_j =\frac 12 \left(
n^2- \sum_{i=1}^u q_i^2 \right)\,, \label{sig1}\\
\sigma_2^\lambda &= \sum_{i=1}^u \left[\frac {q_i}2\right]=
\frac {n-d(q)}2 \,, \label{sig2}\\
\sigma_3^\lambda &= 0\,. \label{sig3}
\end{align}
On the right-hand side of formula (\ref{sig2}), the function
$d(q)$ is defined as the number of odd $q_i$, $1\leq i\leq u$ or,
equivalently:
\[
d(q)= \frac 12 \left(u- \sum_{i=1}^u (-1)^{q_i}\right)\,.
\]
\end{lem}

\begin{proof}
We shall outline the scheme of the proof for the case $u=2$, i.e.,
$q=(q_1, q_2)$, with $q_1+q_2=n$. The generalization to the case
of arbitrary $u$ should be obvious.

The decomposition (\ref{dec}) of $\rn^n$ for $u=2$ becomes $\rn^n
= L_1 \oplus L_2$, with $\dim L_1= q_1$, $\dim L_2= q_2$.
According to this decomposition, the matrix $P^L(m)\in so(n)$ can
be represented in a blockwise form as
\begin{equation*}
P^L=\begin{pmatrix} A_1 & B \\
-\tilde B & A_2
\end{pmatrix}\,,
\end{equation*}
where $A_1\in so(q_1)$, $A_2\in so(q_2)$, whereas $B$ is a generic
$q_1 \times q_2$ matrix. According to remark \ref{typ}, for a
generic $m\in M$ both matrices $A_1$ and $A_2$ will have pairwise
different eigenvalues. By applying lemma \ref{skew} to the two
subspaces $L_1$ and $L_2$, it is possible to find a basis in each
of them such that both matrices $A_1$ and $A_2$ have the normal
block-diagonal form (\ref{block}). In these coordinates, any other
arbitrary matrix $Z\in so(n)$ can be represented as
\begin{equation} \label{zn}
Z=\begin{pmatrix}V_1 & U \\
-\tilde U & V_2
\end{pmatrix} \,,
\end{equation}
where $V_1\in so(q_1)$, $V_2\in so(q_2)$, whereas $U$ is a generic
$q_1 \times q_2$ matrix. We have $Z\in \galg^\lambda$ if and only
if $U=0$. For the commutator of $P^L$ and $Z$ we obtain
\begin{equation} \label{pzn}
[P^L,Z]=\begin{pmatrix}C_1 & D \\
-\tilde D & C_2
\end{pmatrix} \,,
\end{equation}
where
\begin{align}
C_1 &=[A_1 ,V_1]+U\tilde B -B \tilde U \,,\label{c1}\\
C_2 &=[A_2 ,V_2]+\tilde U B -\tilde B U \,,\label{c2}\\
D&= A_1U -UA_2 -V_1 B +BV_2 \,. \label{d}
\end{align}
We will have $\pop^\lambda([P^L,Z]) =0$ if and only if $C_1=0$ and
$C_2= 0$.

Let us now consider the number $\sigma_2^\lambda$ defined by
formula (\ref{sigma2}). The kernel of the operator $\pop^\lambda
\circ {\rm ad}_{P^L}|_{\galg^\lambda}$ is made by the matrices $Z$
of the form (\ref{zn}), with $U=0$, such that $C_1=0$ and $C_2= 0$
in (\ref{pzn}). Using formulas (\ref{c1})--(\ref{c2}), we find
that $V_1$ and $V_2$ must satisfy the conditions $[A_1 ,V_1]=0$
and $[A_2 ,V_2]=0$ respectively. According to formula (\ref{kera})
we have that $\dim{\rm Ker}\, {\rm ad}_{A_1}= [q_1/2]$ and
$\dim{\rm Ker}\, {\rm ad}_{A_2}= [q_2/2]$, whence
\[
\sigma_2^\lambda= \left[\frac {q_1}2\right] +\left[\frac
{q_2}2\right]\,,
\]
which corresponds to (\ref{sig2}).

In a similar way, in order to evaluate the number
$\sigma_3^\lambda$ defined by formula (\ref{sigma3}), we observe
that the kernel of the operator ${\rm ad}_{P^L}|_{\galg^\lambda}$
is made by the matrices $Z$ of the form (\ref{zn}), with $U=0$,
such that $[P^L,Z]=0$. Using formulas (\ref{pzn})--(\ref{d}), we
find that $V_1$ and $V_2$ must simultaneously satisfy the
conditions
\begin{equation}
[A_1 ,V_1]=0\,, \qquad [A_2 ,V_2]=0\,, \label{a1v}
\end{equation}
and
\begin{equation}
V_1 B- BV_2=0\,. \label{v1b}
\end{equation}
We recall that $A_1$ and $A_2$ are skew-symmetric matrices in the
normal block-diagonal form (\ref{block}), with pairwise different
eigenvalues. Hence, according to lemma \ref{skew}, conditions
(\ref{a1v}) imply that also $V_1$ and $V_2$ must have the normal
block-diagonal form (\ref{block}), with arbitrary eigenvalues. But
then it is easy to verify that, for a generic $B$, condition
(\ref{v1b}) necessarily implies that all eigenvalues of $V_1$ and
$V_2$ must be zero. From this we conclude that $Z=0$, so that
(\ref{sig3}) is proved.

Finally, in order to evaluate the number $\sigma_1^\lambda$
defined by formula (\ref{sigma1}), we note that the kernel of the
operator $\pop^\lambda \circ {\rm ad}_{P^L}$ is made by the
matrices $Z$ of the form (\ref{zn}), such that $C_1=0$ and $C_2=
0$ in (\ref{pzn}). Introducing the matrices $T_1:= B \tilde U
-U\tilde B$ and $T_2:= \tilde B U- \tilde U B$, for such $Z$
formulas (\ref{c1})--(\ref{c2}) can be rewritten as
\begin{equation} \label{avt}
[A_1,V_1]= T_1\,, \qquad [A_2,V_2]= T_2\,.
\end{equation}
Hence $U$ must be such that
\begin{equation} \label{st0}
T_1\in{\rm ad}_{A_1}(\galg_1)\,, \qquad T_2\in {\rm
ad}_{A_2}(\galg_2)\,,
\end{equation}
where $\galg_1= so(q_1)$ and $\galg_2= so(q_2)$. According to
formula (\ref{ada}), this is equivalent to
\begin{equation}\label{st}
\begin{split}
(T_1)_{12}=(T_1)_{34}=\dots =(T_1)_{2s_1-1,2s_1}&=0 \,, \\
(T_2)_{12}=(T_2)_{34}=\dots =(T_2)_{2s_2-1,2s_2}&=0 \,,
\end{split}
\end{equation}
where $s_1=[q_1/2]$, $s_2=[q_2/2]$. It is easy to check that, for
a typical matrix $B$, the linear system of equations (\ref{st})
can be solved with respect to $s_1+ s_2$ appropriately chosen
elements of the matrix $U$. Therefore matrices $U$ satisfying
conditions (\ref{st0}) form a linear space having dimension $q_1
q_2- s_1 -s_2$. After choosing $U$ in this space, in order to
obtain an element $Z\in \ker \left(\pop^\lambda \circ {\rm
ad}_{P^L}\right)$ one has to take $V_1$ and $V_2$ satisfying
equations (\ref{avt}). These are equivalent to
\[
V_1\in {\rm ad}_{A_1}^{-1}(T_1)\,, \qquad V_2\in {\rm
ad}_{A_2}^{-1}(T_2)\,.
\]
According to corollary \ref{n/2}, we have
\begin{align*}
\dim {\rm ad}_{A_1}^{-1}(T_1)= \dim\ker {\rm ad}_{A_1}&= s_1\,, \\
\dim {\rm ad}_{A_2}^{-1}(T_2)= \dim\ker {\rm ad}_{A_2}&= s_2\,.
\end{align*}
More exactly, equations (\ref{avt}) determine all elements of
$V_1$ and $V_2$ except $(V_1)_{12}$, $(V_1)_{34}$, $\dots$,
$(V_1)_{2s_1-1,2s_1}$, and $(V_2)_{12}$, $(V_2)_{34}$, $\dots$,
$(V_2)_{2s_2-1,2s_2}$. These are $s_1+s_2$ elements which can be
arbitrarily chosen, in addition to $q_1 q_2- s_1 -s_2$ elements of
the matrix $U$ previously considered. We thus conclude that
\[
\sigma_1^\lambda= q_1q_2 \,,
\]
which corresponds to (\ref{sig1}).
\end{proof}

Let $C^j$ denote the standard set of Casimir functions for the
adjoint action of the group $SO(q_j)$ of orthogonal
transformations on euclidean subspace $L_j \subseteq \rn^n$, $j=1,
\dots, u$, see definition \ref{defcas}. This set contains
$\left[q_j/2\right]$ functions on the algebra $so(q_j)$ which
corresponds to the subspace $L_j\subseteq \rn^n$. We can consider
these functions as functions on the algebra $\galg =so(n)$. Let us
consider the set $C^\lambda =(C^1, \dots, C^u)$ of functions on
$\galg$, obtained by collection of sets $C^j$. Clearly $C^\lambda$
contains
\begin{equation}\label{slam}
s^\lambda:= \sum_{j=1}^u \left[\frac {q_j}2 \right]= \frac
{n-d(q)}2
\end{equation}
elements. Let $C^{L\lambda}= (C^{L1}, \dots, C^{Lu})$ denote the
set $C^{L\lambda} :=C^\lambda \circ P^L$ of functions on $T^*G$,
obtained by making the composition of the functions of set
$C^\lambda$ with the map $P^L$. We have already introduced the set
$\bar C= C\circ P^L= C\circ P^R$ (see proposition \ref{coin}).
Clearly $C^{L\lambda}= \bar C$ if $u=1$. We will denote with
$Z^\lambda$ the set of functions $Z^\lambda:= (\bar C,
C^{L\lambda})$, if $u>1$, or $Z^\lambda:= \bar C$, if $u=1$. The
set $Z^\lambda$ contains $z^\lambda$ elements, where
\begin{equation}\label{zlam}
z^\lambda= \begin{cases} [n/2] \qquad &\text{if } u=1 \,, \\
[n/2] +s^\lambda.
\end{cases}
\end{equation}

\begin{prop} \label{zlamb}
For any $\lambda= (\lambda_1, \dots, \lambda_n)$, each function of
the set $Z^\lambda$ is in involution with each function of the set
$B^\lambda$:
\begin{equation} \label{zb}
\{Z^\lambda, B^\lambda\}=0 \,.
\end{equation}

Furthermore, almost everywhere in $M=T^* G$, where $G=SO(n)$, the
set $Z^\lambda$ is functionally independent, i.e., $\rank
Z^\lambda =z^\lambda$.
\end{prop}

\begin{proof}
The equality $\{\bar C, B^\lambda\}=0$ follows from (\ref{cb}).
From the definition of Casimir function, and the fact that $P^L$
is a Poisson map, it follows that $\{C^{Li},P^{Li}\}=0\, \forall
\,i=1, \dots, u$, where $P^{Li}:= \{P^L_{hk}: (h,k)\in I^i\}$, see
formula (\ref{ilam}). Relation (\ref{plpl}) implies
$\{P^{Li},P^{Lj}\}=0\, \forall \,i\neq j$, whence $\{C^{Li},
P^{Lj}\} =0$. Furthermore, (\ref{prpl}) implies
$\{C^{L\lambda},P^R\}=0$. Hence $\{C^{L\lambda}, B^\lambda\}=0$,
and equality (\ref{zb}) is proved.

According to proposition \ref{r1}, applied to the group $SO(n)$
and its subgroups $SO(q_j)$, $j=1, \dots, u$, both sets $\bar C$
and $C^{L\lambda}$ are almost everywhere functionally independent,
so that
\begin{equation}\label{zl1}
\rank \bar C= \ipn2\,, \qquad \rank C^{L\lambda} = s^\lambda\,.
\end{equation}
It follows that
\begin{align}\label{rzl1}
\rank Z^\lambda &= \dim \sspan(d\bar C, dC^{L\lambda}) \nonumber
\\
&= \ipn2 +s^\lambda -\dim (\sspan d\bar C \cap \sspan
dC^{L\lambda})\,.
\end{align}

Let us then consider the set
\begin{equation}\label{w1}
G =\Big\{w\in T^*_m M:w=\sum_{i<j} K_{ij} dP^L_{ij}\,,\, K\in {\rm
Ker} \left ({\rm ad}_{P^L}|_{\galg^\lambda} \right) \Big\}\,.
\end{equation}
Using (\ref{kpr}) and (\ref{dbskew}) we see that
\begin{equation}\label{w2}
G= \sspan d\bar C \cap \sspan dP^{L\lambda}\,.
\end{equation}
Since $\sspan dC^{L\lambda} \subseteq \sspan dP^{L\lambda}$, we
obtain
\begin{equation}\label{w3}
G\supseteq \sspan d\bar C \cap \sspan dC^{L\lambda}\,.
\end{equation}
On the other hand, any $w\in G$ can be expressed as
$w=\sum_{i=1}^u w_i$, with $w_i\in \sspan dP^{Li}$. For any $z\in
\sspan dP^{Li}$ we have $0=\Pi(w,z) =\Pi(w_i,z)$. Hence $w_i \in
\sspan dP^{Li} \cap(\sspan dP^{Li})^\angle= \sspan dC^{Li}$, where
the last equality follows from the application of (\ref{spanc}) to
the subspace $L_i$. It follows that $w\in \sspan dC^{L\lambda}$.
Hence, using (\ref{w2}) and (\ref{w3}), we conclude that
\[
G= \sspan d\bar C \cap \sspan dC^{L\lambda}\,,
\]
so that
\begin{equation}\label{zl2}
\dim \left({\rm Span}d \bar C \cap {\rm Span}dC^{L\lambda} \right)
= \dim G= \dim(\ker {\rm ad}_{P^L}|_{\galg^\lambda})=
\sigma^\lambda_3\,.
\end{equation}
Recalling formula (\ref{sig3}), and the equalities $s^\lambda=
\sigma_3^\lambda= [n/2]$ for $u=1$, we then obtain from
(\ref{rzl1}) that $\rank Z^\lambda = z^\lambda$ for any $u$, so
that the set $Z^\lambda$ is functionally independent.
\end{proof}

\begin{lem} \label{defect2}
Almost everywhere in $T^*G$ we have
\begin{equation}\label{zblam}
\sspan dZ^\lambda= \sspan dB^\lambda \cap (\sspan
dB^\lambda)^\angle\,.
\end{equation}
Moreover, if $u>1$ we have
\begin{align}
\rank B^\lambda &= 2N- \sum_{i<j} q_i q_j= \frac 12 \left(
n^2-2n +\sum_{i=1}^u q_i^2\right)\,, \label{rblam2}\\
k(B^\lambda) &= z^\lambda= \ipn2+ \sum_{i=1}^u \left[\frac
{q_i}2\right] = n-\left[\frac{d(q)+1}2\right]
\,, \label{kblam2}\\
r(B^\lambda) &= \sum_{i<j} q_i q_j- z^\lambda= \frac 12 \left(
n^2-2n-\sum_{i=1}^u q_i^2 \right)+\left[\frac{d(q)+1}2\right]
\label{dblam2}\,,
\end{align}
where $k(B^\lambda)$ and $r(B^\lambda)$ are respectively the
centrality and the defect of integrability (see definition
\ref{kr}) of the set $B^\lambda$.
\end{lem}

\begin{proof}
Since the elements of $Z^\lambda$ are functions of $B^\lambda$,
from equality (\ref{zb}) it follows that $\sspan dZ^\lambda
\subseteq W:= \sspan dB^\lambda \cap (\sspan dB^\lambda)^\angle$.
From (\ref{dimw2}), taking into account (\ref{sign2}),
(\ref{sig2}) and (\ref{sig3}), we obtain that $\dim W= z^\lambda$.
Moreover, we know from proposition \ref{zlamb} that the set
$Z^\lambda$ is functionally independent, so that $\dim (\sspan
dZ^\lambda)= z^\lambda$. Hence we conclude that $\sspan
dZ^\lambda= W$, so that (\ref{zblam}) is proved.

Equality (\ref{rblam2}) follows from (\ref{rblam}) and
(\ref{sig1}).

According to definition \ref{kr}, the centrality of $B^\lambda$ is
at least as great as the number of elements of $Z^\lambda$, i.e.,
we have $k(B^\lambda)\geq z^\lambda$. On the other hand, using
(\ref{sign2}), (\ref{sig2}) and (\ref{sig3}), we obtain from
(\ref{kblam}) that $k(B^\lambda) \leq z^\lambda$. Therefore
(\ref{kblam2}) is proved.

Finally, (\ref{dblam2}) follows from definition \ref{kr} and from
equalities (\ref{rblam2})--(\ref{kblam2}).
\end{proof}
Note that one can easily prove, by induction on $u$, that
$r(B^\lambda)$ given by formula (\ref{dblam2}) is always an even
integer.

The complete integrability of the system describing the free
rotation of an $n$-dimensional rigid body can be proved by
introducing the so-called Manakov's integrals \cite{man}. It is
not difficult to check that the function ${\cal P}(\rho) \equiv
(1/2k) \Tr (P^L +J^2 \rho)^k: T^*G \to \rn$, where $J$ is the
diagonal $n\times n$ matrix with $\lambda_1, \lambda_2, \dots,
\lambda_n$ as diagonal elements, is in involution with $H^\lambda$
for any value of the parameter $\rho$. Hence the coefficients of
the polynomial ${\cal P}(\rho)$ in the variable $\rho$ are also in
involution with $H^\lambda$, i.e.\ we have $\{c_{ij},
H^\lambda\}=0$, where
\[
{\cal P}(\rho) \equiv\frac 1{2k} \Tr (P^L +J^2 \rho)^k=
\sum_{j=0}^{k} c_{kj} \rho^{j}\,.
\]
These coefficients are not all functionally independent. One
immediately sees that $c_{kk}$ is just a constant, and that
$c_{kj}=0$ whenever $k-j$ is odd. Furthermore it can be proved
that, if one is only interested to functionally independent
elements, then one need only consider coefficients $c_{kj}$ with
$k=2, 3, \dots, n$. It is easy to see that ${\cal P}(\rho)$ is in
involution with all functions of the set $B^\lambda$, so that
$\{c_{ij}, B^\lambda\} = 0$. Moreover, one can prove that all
coefficients $c_{ij}$ are mutually in involution, $\{c_{ij},
c_{i'j'}\}=0$. This result provides in particular another proof of
the fact that all these coefficients are integrals of the system
with hamiltonian $H^\lambda$, for it can be shown that $H^\lambda$
can be expressed as a linear combination of the functions
$c_{k,k-2}$ for $k=2,\dots, n$.

It has been proved in general that the system with hamiltonian
$H^\lambda$ is integrable \cite{mish2, mish, ratiu}. In the
general case in which all generalized moments of inertia are
pairwise different, $\lambda_i\neq \lambda_j$ for $i \neq j$, an
integrable set of functions is given by $(M; P'^R)$, where
$M=(c_{k,k-2i}, k=2, \dots, n,\ i=1, 2,\dots, [k/2])$ is the
complete set of $\sum_{k=2}^n [k/2] =(1/2)(n(n-1)/2 +[n/2])$
functionally independent coefficients $c_{ij}$, and $P'^R$ is a
set of $n(n-1)/2 -[n/2]$ elements of $P^R$, such that $(\bar C,
P'^R)$ is a functionally independent set. Hence this system is
integrable with $(1/2)(n(n-1)/2 +[n/2])$ central functions. The
elements $c_{ij}$ of $M$ such that $j>0$ are called Manakov's
integrals. The remaining $[n/2]$ elements of $M$, i.e.\ $c_{2k,0}
=(1/2k)\Tr (P^L)^k$ for $k=1,2, \dots, [n/2]$, are independent of
the moments of inertia $\lambda$, and form a set of Casimir
functions equivalent to the set $\bar C$ introduced in section
\ref{fr}. Hence, an equivalent integral set of functions for the
free $n$-dimensional rigid body with pairwise different moments of
inertia is $(\bar C, \bar M; P'^R)$, where $\bar M$ is the set of
$(1/2)(n(n-1)/2 -[n/2])$ Manakov's integrals.

When the moments of inertia are not all pairwise different, the
set $M$ is no longer functionally independent. However the
integrability of the system is preserved, which means that one can
construct an integrable set of functions whose central subset is
made of the elements of $Z^\lambda$ and of a suitable subset of
$\bar M$. According to proposition \ref{strind2} such a subset of
$\bar M$ must contain just $r/2$ elements, where $r=r(B^\lambda)$
is the defect of integrability of the set $B^\lambda$ and is given
by formula (\ref{dblam2}). Therefore, the central subset will
contain $k(B^\lambda)+ r(B^\lambda)/2$ elements. This result is
expressed by the following proposition.

\begin{prop}\label{rbint}
The system with hamiltonian $H=H^\lambda$ given by formula
(\ref{hamlam}) is integrable with
\begin{equation}\label{kq}
\bar k(q)= \frac 14 \left( n^2+2n-\sum_{i=1}^u q_i^2 \right)-
\frac 12 \left[\frac{d(q)+1}2\right]
\end{equation}
central integrals.
\end{prop}

Manakov's integrals can be explicitly represented in the following
form:
\begin{equation}\label{man}
c_{k,k-2l}=\frac 1{4l}\sum_{i_1,i_2,\dots, i_{2l}} a_{k,k-2l}
^{i_1 i_2 \dots i_{2l}} P^L_{i_1 i_2} P^L_{i_2 i_3} \cdots
P^L_{i_{2l-1} i_{2l}} P^L_{i_{2l} i_1}\,,
\end{equation}
with $0<l<k/2$, where
\begin{equation} \label{akl}
a_{k,k-2l} ^{i_1 i_2 \dots i_{2l}}= \sum_{b_1\geq 0, b_2\geq 0,
\dots, b_{2l}\geq 0} \lambda_{i_1}^{2b_1} \lambda_{i_2}^{2b_2}
\cdots \lambda_{i_{2l}}^{2b_{2l}}\delta_{b_1 +b_2+ \cdots +b_{2l},
k-2l}\,.
\end{equation}
We see that $c_{k,k-2l}$ is a homogeneous polynomial of degree
$2l$ in the left-invariant momenta, while its coefficients
$a_{k,k-2l} ^{i_1, \dots, i_{2l}}$ are homogeneous polynomials of
degree $2(k-2l)$ in the generalized moments of inertia, completely
symmetrical with respect to permutations of the indexes $i_1,
\dots, i_{2l}$.

In Table \ref{tabrb} we give the number $\bar k(q)$ of central
integrals resulting from the above proposition for free
$n$-dimensional rigid bodies with $n\leq 6$. We also give the
quantities $k(B^\lambda)$ and $r(B^\lambda)$ resulting from lemma
\ref{defect2}.
\begin{table}[h]
\centering
\begin{tabular}{l|l|c|c|c }
$n$ & $q$ & $k(B^\lambda)$ & $r(B^\lambda)$& $\bar k(q)$\\
\hline
3 & (3) & 1 & 0 & 1 \\
3 & (1,2)& 2 & 0 & 2 \\
3 & (1,1,1) & 1 & 2 & 2\\
4 & (4) & 2 & 0 & 2 \\
4 & (1,3) & 3 & 0 & 3 \\
4 & (2,2) & 4 & 0 & 4 \\
4 & (1,1,2) & 3 & 2 & 4 \\
4 & (1,1,1,1) & 2 & 6 & 5 \\
5 & (5) & 2 & 0 & 2 \\
5 & (1,4) & 4 & 0 & 4 \\
5 & (2,3) & 4 & 2 & 5 \\
5 & (1,1,3) & 3 & 4 & 5 \\
5 & (1,2,2) & 4 & 4 & 6 \\
5 & (1,1,1,2) & 3 & 6 & 6 \\
5 & (1,1,1,1,1) & 2 & 8 & 6 \\
6 & (6) & 3 & 0 & 3 \\
6 & (1,5) & 5 & 0 & 5 \\
6 & (2,4) & 6 & 2 & 7 \\
6 & (3,3) & 6 & 2 & 7 \\
6 & (1,1,4) & 5 & 4 & 7 \\
6 & (1,2,3) & 5 & 6 & 8 \\
6 & (1,1,1,3) & 4 & 8 & 8 \\
6 & (1,1,2,2) & 5 & 8 & 9 \\
6 & (1,1,1,1,2) & 4 & 10 & 9 \\
6 & (1,1,1,1,1,1) & 3 & 12 & 9
\end{tabular}
\caption{\label{tabrb} Number of central integrals for free rigid
bodies.}
\end{table}

\subsection{Free rotation of a quantum rigid body}

In order to quantize a free rigid body we have to consider the
quantum impulses $\hat P^L_{ij}$, $\hat P^R_{ij}$, which are
constructed according to formula (\ref{diffp}) in correspondence
with vector fields $V^L_{ij}$, $V^R_{ij}$ respectively, $1\leq i<j
\leq n$. However, if we want to apply to this system the concept
of integral quantum system introduced in %section \ref{iqs},
\cite{part1}, we are apparently faced by the problem that here the
configuration space $K=G= SO(n)$ is not a domain of the linear
space $\rn^N$. This problem is solved by the consideration of
local coordinates on $G$. Note that the main part $M\fop$ of a
linear differential
operator $\fop$ (see definition %\ref{hg})
in \cite{part1}) is not defined intrinsically. However the symbol
$(M\fop)^{\rm smb}$ can be considered as intrinsically defined,
according to the following proposition.

\begin{prop}
The symbol $S:=(M\fop)^{\rm smb}$ of the main part of a linear
operator of class $\oop$, expressed via local coordinates $x$ on
configuration space $K$, has the form of a homogeneous polynomial
of $p$. This polynomial $S=S(x,p)$ behaves under a change of local
coordinates $x$ on $K$ as a function on the cotangent bundle
$T^*K$ to the manifold $K$. It follows from this fact that the
definition of quasi-independence of a set of operators does not
depend on the choice of local coordinates on configuration space
$K$. The same is true for the definition of quasi-integrability of
either a set of operators or an individual operator.
\end{prop}
The proof of this proposition is obvious. The first part of the
proposition, about the representation of $S$ as a function on
$T^*K$, is actually the reformulation of well-known facts.

Let us consider the quantum system with hamiltonian operator
\begin{equation} \label{qhamlam}
\hat H= \hat H^{\lambda}= \frac 12 \sum_{i<j} \frac {\left(\hat
P^L_{ij}\right)^2}{\lambda_i +\lambda_j} %\,.
\end{equation}
on $C^\infty (SO(n))$. We consider this system as the system
describing the free rotation of a quantum $n$-dimensional rigid
body.

\begin{prop}
For $n\leq 6$ this quantum system is quasi-integrable for any
$\lambda$, with the same number $\bar k$ of central operators as
the number $\bar k(u)$ of central integrals of the corresponding
classical system, see proposition \ref{rbint} and Table
\ref{tabrb}.

Moreover, if $q=(n)$, this quantum system is quasi-integrable for
any $n$ with $[n/2]$ central operators. If $q=(1,n-1)$, this
quantum system is quasi-integrable for any $n$ with $n-1$ central
operators.
\end{prop}

\begin{proof}
Let $\hat P^L$ and $\hat P^R$ denote the sets of operators $\hat
P^L=(\hat P^L_{ij}, 1\leq i<j \leq n)$ and $\hat P^R=(\hat
P^R_{ij}, 1\leq i<j \leq n)$. Let us consider the set $\hat B:=
(\hat P^L, \hat P^R)$ containing $2N$ operators. Let $\hat C^R$
denote the set of operators which are obtained by symmetrization
with respect to $\hat P^R$ from the functions of set $C^R=
C(P^R)$, i.e., $\hat C^R:= (C^R)^ {\rm sym}$. Analogously we
define $\hat C^L$, $\hat P^{L\lambda}$, $\hat B^{\lambda}$, $\hat
C^{L\lambda}$, $\hat Z^{\lambda}$, i.e., $\hat C^L:= (C^L)^ {\rm
sym}$ etc. We first show that the commutation relations $[\hat
B^\lambda, \hat H^\lambda]= [\hat B^\lambda, \hat Z^\lambda]=0$
follow from the analogous classical relations $\{B^\lambda,
H^\lambda\}= \{B^\lambda, Z^\lambda\}=0$ and from the propositions
about quantization of %sections \ref{apq}
\cite{part2} and section \ref{qcff} of the present paper. The
commutators of the operators of set $\hat B$ have the same form as
the Poisson brackets of the corresponding classical functions,
since all these functions are linearly dependent on canonical
impulses $p$, see formulas (\ref{iplpr})--(\ref{prpl}) and
proposition \ref{isom}:
\begin{align}
[\hat P_{ij}^L, \hat P_{hk}^L] &= -\delta_{ih} \hat P^L_{jk} -
\delta_{jk} \hat P^L_{ih} + \delta_{ik} \hat P^L_{jh}
+ \delta_{jh} \hat P^L_{ik} \,, \label{qplpl}\\
[\hat P_{ij}^R, \hat P_{hk}^R] &= \delta_{ih} \hat P^R_{jk} +
\delta_{jk} \hat P^R_{ih} - \delta_{ik} \hat P^R_{jh}
- \delta_{jh} \hat P^R_{ik} \,, \label{qprpr}\\
[\hat P_{ij}^L, \hat P_{hk}^R] &= 0 \,. \label{qprpl}
\end{align}
The relation $[\hat H^{\lambda}, \hat B^{\lambda}] =0$ then
follows from $\{H^\lambda, B^\lambda\}=0$ using proposition
%\ref{appl2},
4.2, case b, of \cite{part2}. Similarly, since
\[
\{C^L(P^L),B^\lambda\}=0 \,, \qquad \hat C^L= (C^L)^ {\rm sym} \,,
\]
the equality $[\hat C^L, \hat B^\lambda]=0$ is obtained by
applying corollary %\ref{casim2}.
4.4 of \cite{part2}. Analogously, we obtain the equality $[\hat
C^{L\lambda}, \hat B^{\lambda}] =0$ from the corresponding
classical relation $\{C^{L\lambda}, B^{\lambda}\} =0$. We have
thus proved that $[\hat B^\lambda, \hat Z^\lambda] =0$.

If $q=(n)$ or $q=(1,n-1)$, according to lemma \ref{defect2} the
defect of integrability of the set $B^\lambda$ in the classical
case is $r(B^\lambda)=0$. Moreover, we have $k(B^\lambda)=[n/2]$
if $q=(n)$, and $k(B^\lambda)=n-1$ if $q=(1,n-1)$. This implies
that in these two cases, for any $n$, there exists a classical
integrable set of functions of the form $F=(Z^\lambda; B')$, where
$B'\subset B^\lambda$, and the central subset $Z^\lambda$ contains
$k(B^\lambda)$ elements. Let us then consider the corresponding
set of operators $\hat F=(\hat Z^\lambda; \hat B')$. From what we
have seen above, it follows that $[\hat Z^\lambda, \hat F]=0$.
Moreover, since all functions of $F$ are homogeneous with respect
to $p$, these functions coincide with their main parts with
respect to $p$, i.e., $M(F)= F$. It is also easy to see that the
elements of $F$ are the symbols of the main parts with respect to
$\hat p$ of the elements of the corresponding set of operators
$\hat F$, i.e., $F= (M\hat F)^{\rm smb}$. Hence, the
quasi-independence of the sets of operators $\hat F$ follows
immediately from the functional independence of set of functions
$F$. One thus concludes that $\hat F$ is an integrable set of
operators with $k(B^\lambda)$ central elements. One can also
easily show that in these two cases, $\hat H^\lambda$ is a linear
combination of the elements of $\hat Z^\lambda$. The integrability
of the system describing the free quantum rigid-body is thus
proved for any $n$ in the two cases $q=(n)$ and $q=(1,n-1)$.

In the remaining cases, the classical integrable sets of functions
generally include also one or more Manakov's integrals among their
central elements. We define Manakov's operators $\hat c_{k,k-2l}$
as the symmetrization of the classical functions (\ref{man}) with
respect to the left-invariant momenta:
\begin{equation}
\hat c_{k,k-2l}= \frac 1{4l}\sum_{i_1,i_2,\dots, i_{2l}}
a_{k,k-2l} ^{i_1 i_2 \dots i_{2l}} {\rm Sym}_{2l}(\hat P^L_{i_1
i_2}, \hat P^L_{i_2 i_3}, \dots, \hat P^L_{i_{2l-1} i_{2l}}, \hat
P^L_{i_{2l} i_1})\,,
\end{equation}
with $0<l<k/2$. Note that, also in the quantum case, the
hamiltonian operator $\hat H^\lambda$ can be expressed as a linear
combination of the operators $\hat c_{k,k-2}$ for $k=2,\dots, n$.
By applying again the results of \cite{part2} we easily see that
$[\hat c_{k,k-2l}, \hat B^\lambda]=0$. However, no general theorem
ensures that Manakov's operators commute with $\hat H^\lambda$ or
among themselves. We will here limit ourselves to studying
commutators between Manakov's operators of degree lower that 6 in
the momenta. This will be sufficient to establish the
quasi-integrability of the free quantum rigid body in spatial
dimensions $n\leq 6$.

The commutator between two Manakov's operators, when one of them
is of second degree, can be evaluated by making use of proposition
%\ref{pc2}
2.3 of \cite{part2}, and of the algebra (\ref{qplpl}) of
left-invariant momenta. In this way, with some computation we find
for any $l,h$:
\begin{align}
[\hat c_{l,l-2}, \hat c_{h,h-2}] &=0 \,, \label{c2c2}\\
[\hat c_{l,l-2}, \hat c_{h,h-4}] &= \frac 16\sum_{i,j,k}
b^{ijk}_{l,h} {\rm Sym}_{3}(\hat P^L_{i j}, \hat P^L_{j k}, \hat
P^L_{k i})\,, \label{c2c4}
\end{align}
where
\begin{align}
b^{ijk}_{l,h}= &\ a^{ij}_{l,l-2} \Big(2a_{h,h-4}^{iijk}-
3a_{h,h-4}^{iikk} - \sum_{p\neq i,j,k}
a_{h,h-4}^{iikp}\Big) \nonumber\\
&+\sum_{p\neq i,j,k} a^{kp}_{l,l-2} \big(a_{h,h-4}^{iijk} -
a_{h,h-4}^{iijp}\big)\,. \label{blh}
\end{align}
Note that the operator ${\rm Sym}_{3}(\hat P^L_{i j}, \hat P^L_{j
k}, \hat P^L_{k i})$ is completely antisymmetrical with respect to
permutations of indexes $i,j,k$, so that formula (\ref{c2c4}) can
be rewritten as
\[
[\hat c_{l,l-2}, \hat c_{h,h-4}] = \sum_{i<j<k} b^{[ijk]}_{l,h}
{\rm Sym}_{3}(\hat P^L_{i j}, \hat P^L_{j k}, \hat P^L_{k i})\,,
\]
where $b^{[ijk]}_{l,h}$ denotes the complete antisymmetrization of
coefficient $b^{ijk}_{l,h}$.

In formula (\ref{blh}) the coefficients $a^{ij}_{l,l-2}$ and
$a^{ijkp}_{h,h-4}$ have to be replaced by their explicit
expressions given by (\ref{akl}). We have
\begin{equation}\label{al2}
a^{ij}_{l,l-2}= \sum_{k=0}^{l-2} \lambda_i^{2(l-2-k)}
\lambda_j^{2k}= \frac{\lambda_i^{2(l-1)}- \lambda_j^{2(l-1)}}
{\lambda_i^{2}- \lambda_j^{2}}\,.
\end{equation}
Moreover, for $h=5$ we have
\[
a^{ijkp}_{h,h-4}= a^{ijkp}_{5,1} = \lambda_i^{2}+ \lambda_j^{2}+
\lambda_k^{2}+ \lambda_p^{2}\,,
\]
so that from (\ref{blh}) we easily obtain $b^{[ijk]}_{l,5}=0$,
which implies
\begin{equation}
[\hat c_{l,l-2}, \hat c_{5,1}] =0
\end{equation}
and consequently $[\hat H^\lambda, \hat c_{5,1}] =0$. From these
results it follows that the free quantum rigid body is a
quasi-integrable system for spatial dimensions $n\leq 5$. The
integrable set of functions of the classical system can in fact be
quantized by replacing Manakov's integrals with the corresponding
operators. Since these operators are homogeneous in the momenta,
the symbols of their main parts with respect to $\hat p$ coincide
with the corresponding classical functions. Hence the
quasi-independence of the sets of operators is a consequence of
the functional independence of the classical sets of functions.

For $h=6$ the coefficient $a^{ijkp}_{h,h-4}$ becomes
\[
a^{ijkp}_{6,2} = \lambda_i^{4}+ \lambda_j^{4}+ \lambda_k^{4}+
\lambda_p^{4}+ \lambda_i^{2}\lambda_j^{2} +
\lambda_i^{2}\lambda_k^{2}+ \lambda_i^{2}\lambda_p^{2}+
\lambda_j^{2}\lambda_k^{2}+ \lambda_j^{2}\lambda_p^{2}+
\lambda_k^{2}\lambda_p^{2}\,,
\]
and from (\ref{blh}) we obtain
\begin{align}
b^{[ijk]}_{l,6}&= \frac 56 \Big[\lambda_k^2(\lambda_j^2 -
\lambda_i^2)a^{ij}_{l,l-2}+ \lambda_i^2(\lambda_k^2 -
\lambda_j^2)a^{jk}_{l,l-2}+ \lambda_j^2(\lambda_i^2 -
\lambda_k^2)a^{ki}_{l,l-2}\Big]\label{bl6}\\
&= \frac 56 \Big[\lambda_i^{2(l-1)}(\lambda_j^{2} -
\lambda_k^{2})+ \lambda_j^{2(l-1)}(\lambda_k^{2} - \lambda_i^{2})+
\lambda_k^{2(l-1)}(\lambda_i^{2} - \lambda_j^{2})\Big]\,.
\nonumber
\end{align}
Since for a generic set $\lambda$ of generalized moments of
inertia the above expression is different from 0, we have that in
general $[\hat c_{l,l-2},\hat c_{6,2}]\neq 0$. Note also that,
putting $l=3/2$ in (\ref{al2}), we get $a^{ij}_{3/2,-1/2} =
1/(\lambda_i +\lambda_j)$, so that we can formally write $\hat
H^\lambda= -\hat c_{3/2,-1/2}$. From the above formulas we thus
directly obtain
\[
[\hat H^\lambda,\hat c_{6,2}]= \sum_{i<j<k} b^{ijk} {\rm
Sym}_{3}(\hat P^L_{i j}, \hat P^L_{j k}, \hat P^L_{k i})\,,
\]
with
\[
b^{ijk}=-\frac 56 \Big[\lambda_i(\lambda_j^{2} - \lambda_k^{2})+
\lambda_j(\lambda_k^{2} - \lambda_i^{2})+ \lambda_k(\lambda_i^{2}
- \lambda_j^{2})\Big]\,.
\]
Hence $\hat c_{6,2}$ is not in general an integral operator of the
quantum system with hamiltonian $\hat H^\lambda$.

By using again proposition %\ref{pc2}
2.3 of \cite{part2} and commutation relations (\ref{qplpl}), one
finds however that, for arbitrary symmetrical coefficients
$\alpha^{ij}=\alpha^{ji}$,
\[
-\frac 14 \sum_{ij} \alpha^{ij}[\hat c_{l,l-2},(\hat P^L_{ij})^2]=
\sum_{i<j<k} \bar b^{ijk}_{l} {\rm Sym}_{3}(\hat P^L_{i j}, \hat
P^L_{j k}, \hat P^L_{k i})\,,
\]
with
\begin{equation}
\bar b^{ijk}_{l}= (\alpha^{jk}- \alpha^{ik})a^{ij}_{l,l-2}
+(\alpha^{ki}- \alpha^{ji})a^{jk}_{l,l-2}  + (\alpha^{ij}-
\alpha^{kj})a^{ki}_{l,l-2}\,.
\end{equation}
One sees immediately that the above expression becomes identical
with (\ref{bl6}) if one chooses $\alpha^{ij}=(5/6)\lambda_i^2
\lambda_j^2$. Hence, if we define the modified Manakov's operator
\[
\hat C_{6,2} \equiv \hat c_{6,2}+ \frac 5{12}\sum_{i<j}
\lambda_i^2 \lambda_j^2 (\hat P^L_{ij})^2\,,
\]
then we get
\[
[\hat c_{l,l-2},\hat C_{6,2}]=0\,, \qquad \qquad [\hat
H^\lambda,\hat C_{6,2}]=0\,.
\]

For $n=6$, let us then consider the set of operators $\hat F$
which is obtained from the classical integrable set of functions
$F$ given by proposition \ref{rbint}, by replacing Manakov's
integrals $c_{l,l-2}$ $(l=3,\dots,6)$, $c_{5,1}$, and $c_{6,2}$,
with the operators $\hat c_{l,l-2}$, $\hat c_{5,1}$, and $\hat
C_{6,2}$ respectively. In order to prove that this set of
operators satisfies the required commutation relations, we still
have to show that
\begin{equation}\label{c56}
[\hat c_{5,1},\hat C_{6,2}]=0\,.
\end{equation}
With the techniques already employed
% based on proposition \ref{pc2},
one obtains that
\begin{align}
\frac 5{12}\sum_{i<j} \lambda_i^2 \lambda_j^2 [\hat c_{5,1},(\hat
P^L_{ij})^2]=& -\frac 56 \sum_{hlm}\lambda_l^4\lambda_m^2
\bigg(\frac 53 {\rm Sym}_{3}(\hat
P^L_{hl}, \hat P^L_{lm}, \hat P^L_{mh}) \nonumber\\
&+\sum_{ij}{\rm Sym}_{5}(\hat P^L_{i j}, \hat P^L_{j h}, \hat
P^L_{hl}, \hat P^L_{lm}, \hat P^L_{mi})\bigg)\,. \label{c56p}
\end{align}
On the other hand, the commutator $[\hat c_{5,1},\hat c_{6,2}]$
between two fourth-order symmetrized polynomials cannot be worked
out with the tools provided in \cite{part2}. By means of a
straightforward and quite heavy calculation, we have verified that
this commutator is indeed just the opposite of the expression
(\ref{c56p}), so that equality (\ref{c56}) actually holds. Of
course, the symbol of the main part of $\hat C_{6,2}$ coincides
with the classical function $c_{6,2}$, so that the
quasi-independence of the set $\hat F$ again follows from the
functional independence of the classical set $F$. We conclude that
$\hat F$ is a quasi-integrable set of operators. The proposition
is thus completely proved.
\end{proof}

We have seen that, for $n\geq 6$, the correct quantization of
Manakov's integral $c_{6,2}$ does not coincide with the
symmetrization of the classical function with respect to the
left-invariant momenta. We have nevertheless provided a recipe to
achieve the quasi-integrability of the free quantum rigid body for
$n=6$ for arbitrary moments of inertia $\lambda$. We can
conjecture that analogous procedures can lead to the
quasi-integrability of this quantum system also for $n>6$.
However, we are unable at the moment to prove that this conjecture
is true.

\end{document}